\documentclass[hidelinks]{article}
\usepackage[letterpaper, portrait, top=1in, left=1in, right=1in, bottom=1in ]{geometry}

\usepackage[utf8]{inputenc} 
\usepackage[T1]{fontenc}    
\usepackage{hyperref}       
\usepackage{url}            
\usepackage{booktabs}       
\usepackage{amsfonts}       
\usepackage{nicefrac}       
\usepackage{microtype}      
\usepackage{cite}
\usepackage{color}
\usepackage{graphicx}
\usepackage{caption}
\usepackage{subcaption}
\usepackage{amsmath}
\usepackage{amsthm}
\usepackage{enumitem}
\usepackage{authblk}

\usepackage[boxed]{algorithm}
\usepackage{graphicx}
\usepackage{capt-of}
\usepackage{algorithmicx}
\usepackage{algpseudocode}

\title{Estimators for Multivariate Information Measures in General Probability Spaces}

\author[1]{Arman Rahimzamani\thanks{armanrz@uw.edu}}
\author[1]{Himanshu Asnani\thanks{asnani@uw.edu}}
\author[2]{Pramod Viswanath\thanks{pramodv@illinois.edu}}
\author[1]{Sreeram Kannan\thanks{ksreeram@uw.edu}}
\affil[1]{Department of Electrical and Computer Engineering, University of Washington, Seattle, WA}
\affil[2]{Department of Electrical and Computer Engineering, University of Illinois at Urbana-Champaign, IL}

\newtheorem{theorem}{Theorem}

\newtheorem{assumption}{Assumption}

\newtheorem{lemma}[theorem]{Lemma}
\newtheorem{proposition}{Proposition}[section]
\newtheorem*{remark}{Remark}

\usepackage{glossaries}
\makeglossaries
\newglossaryentry{pa_xl_var}{
  name={\ensuremath{\text{pa}(X_l)}},
  description={The set of random variables representing the parents of $X_l$},
  sort={L}}
\newglossaryentry{pa_n_xl_var}{
  name={\ensuremath{\text{pa+}(X_l)} },
  description={The set of random variables representing the parents of $X_l$ and itself},
  sort={L}}
\newglossaryentry{pa_xl_val}{
  name={ \ensuremath{x_{\text{pa}(l)}} },
  description={The given value for the set of random variables representing the parents of $X_l$},
  sort={L}}
\newglossaryentry{pa_n_xl_val}{
  name={ \ensuremath{x_{\text{pa+}(l)}} },
  description={The given value for the set of random variables representing the parents of $X_l$ and itself},
  sort={L}}
\newglossaryentry{pa_xl_space}{
  name={ \ensuremath{\mathcal{X}_{\text{pa}(l)}} },
  description={The sample space for a set of random variables representing the parents of $X_l$ and itself},
  sort={L}}
\newglossaryentry{pa_n_xl_space}{
  name={ \ensuremath{\mathcal{X}_{\text{pa+}(l)}} },
  description={The sample space for a set of random variables representing the parents of $X_l$ and itself},
  sort={L}}
\newglossaryentry{graph_dist}{
  name={ \ensuremath{\overline{\mathbb{P}}} },
  description={The distribution over $X$ imposed by the graphical model $G$},
  sort={L}}
\newglossaryentry{graph_distance_measure}{
  name={ \ensuremath{\mathbb{GDM}} },
  description={The distribution over $X$ imposed by the graphical model $G$},
  sort={L}}

\newcommand{\indep}{\rotatebox[origin=c]{90}{$\models$}}
\begin{document}
\maketitle

\begin{abstract}
Information theoretic quantities play an important role in various settings in machine learning, including causality testing, structure inference in graphical models, time-series problems, feature selection as well as in providing privacy guarantees. A key quantity of interest is the mutual information and generalizations thereof, including conditional mutual information, multivariate mutual information, total correlation and directed information. While the aforementioned information quantities are well defined in arbitrary probability spaces, existing estimators add or subtract entropies (we term them $\Sigma H$ methods). These methods work only in purely discrete space or purely continuous case since entropy (or differential entropy) is well defined only in that regime. 

In this paper, we define a general graph divergence measure ($\mathbb{GDM}$),as  a measure of incompatibility between the observed distribution and a given graphical model structure. This generalizes the aforementioned information measures and we construct a novel estimator via a coupling trick that directly estimates these multivariate information measures using the Radon-Nikodym derivative. These estimators are proven to be consistent in a general setting which includes several cases where the existing estimators fail, thus providing the only known estimators for the following settings: (1) the data has some discrete and some continuous valued components (2) some (or all) of the components themselves are discrete-continuous \textit{mixtures} (3) the data is real-valued but does not have a joint density on the entire space, rather is supported on a low-dimensional manifold. We show that our proposed estimators significantly outperform known estimators on synthetic and real datasets. 
\end{abstract}

\section{Introduction}


Information theoretic quantities, such as mutual information and its generalizations, play an important role in various settings in machine learning and statistical estimation and inference. Here we summarize briefly the role of some generalizations of mutual information in learning (cf. Sec.~\ref{sec:math-inf-quantities} for a mathematical definition of these quantities).

\begin{enumerate}[leftmargin=*]
\item {\bf Conditional mutual information} measures the amount of information between two variables $X$ and $Y$ given a third variable $Z$ and is zero iff $X$ is independent of $Y$ given $Z$. CMI finds a wide range of applications in learning including causality testing \cite{dawid1979conditional, zhang2012kernel},   structure inference in graphical models \cite{whittaker2009graphical}, feature selection \cite{fleuret2004fast} as well as in providing privacy guarantees \cite{cuff2016differential}.
\item {\bf Total correlation} measures the degree to which a set of $N$ variables are independent of each other, and appears as a natural metric of interest in several machine learning problems, for example, in independent component analysis, the objective is to maximize the independence of the variables quantified through total correlation \cite{hyvarinen2000independent}. In feature selection, ensuring the independence of selected features is one goal, pursued using total correlation in \cite{vergara2014review,meyer2008information}. 
\item {\bf Multivariate mutual information} measures the amount of information shared between multiple variables \cite{mcgill1954multivariate,chan2015multivariate} and is useful in feature selection \cite{lee2013feature,brown2009new} and clustering \cite{chan2016info}. 
\vspace{-0.05in}\item {\bf Directed information} measures the amount of information between two random processes \cite{watanabe1960information,permuter2011interpretations} and is shown as the correct metric in identifying time-series graphical models \cite{quinn2015directed,sun2015causal,hlavavckova2007causality, amblard2011directed, rahimzamani2016network, rahimzamani2017potential}.
\end{enumerate}

Estimation of these information-theoretic quantities from observed samples is a non-trivial problem that needs to be solved in order to utilize these quantities in the aforementioned applications. While there has been long history in estimation of entropy \cite{kozachenko1987sample, beirlant1997nonparametric, wieczorkowski1999entropy, miller2003new}, and renewed recent interest \cite{lee2010sample,sricharan2013ensemble,han2015adaptive}, much less effort has been spent on the multivariate versions. A standard approach to estimating general information theoretic quantities is to write them out as a sum or difference of entropy (denoted $H$ usually) terms which are then directly estimated; we term such a paradigm as $\Sigma H$ paradigm. However, the $\Sigma H$ paradigm is applicable only when the variables involved are all discrete or there is a joint density on the space of all variables (in which case, differential entropy $h$ can be utilized). The underlying information measures themselves are well defined in very general probability spaces, for example, involving mixtures of discrete and continuous variables; however, no known estimators exist. 

We motivate the requirement of estimators in general probability spaces by some examples in contemporary machine learning and statistical inference. 
\begin{enumerate}[leftmargin=*]
\item It is common place in machine learning to have data-sets where {\bf some variables are discrete, and some are continuous}. For example, in recent work on utilizing information bottleneck to understand deep learning \cite{tishby2015deep}, an important step is to quantify the mutual information between the training samples (which are discrete) and the layer output (which is continuous). The employed methodology was to quantize the continuous variables; this is common practice, even though highly sub-optimal. 
\item Some variables involved in the calculation {\bf may be mixtures of discrete and continuous variables}. For example, the output of ReLU neuron will not have a density even when the input data has a density. Instead, the neuron will have a discrete mass at $0$ (or wherever the ReLU breakpoint is) but will have a continuous distribution on the positive values. This is also the case in gene expression data, where a gene may have a discrete mass at expression $0$ due to an effect called drop-out \cite{liu2016single} but have a continuous distribution elsewhere. 
\item The variables involved may have a joint density {\bf only on a low dimensional manifold}. For example, when calculating mutual information between input and output of a neural network, some of the neurons are deterministic functions of the input variables and hence they will have a joint density supported on a low-dimensional manifold rather than the entire space. 
\end{enumerate}

In the aforementioned cases,  no existing estimators are known to work. It is not merely a matter of having provable guarantees either. When we plug in estimators that assume a joint density into data that does not, the estimated information measure can be strongly negative.

We summarize our main contributions below: 
\begin{enumerate}[leftmargin=*]
\item {\bf General paradigm} (Section \ref{sec:gdm}): We define a general paradigm of graph divergence measures which captures the aforementioned generalizations of mutual information as special cases. Given a directed acyclic graph (DAG) between $n$ variables, the graph divergence is defined as the Kullback-Leibler (KL) divergence between the true data distribution $\mathbb{P}_X$ and a restricted distribution $\gls{graph_dist}_X$ defined on the Bayesian network and can be thought of as a measure of incompatibility with the given graphical model structure. These graph divergence measures are defined using the Radon Nikodym derivatives which are well-defined for general probability spaces. 

\item {\bf Novel estimators} (Section \ref{sec:estimators}): We propose novel estimators for these graph divergence measures, which directly estimate the corresponding Radon-Nikodym derivatives. To the best of our knowledge, these are the first family of estimators that are well defined for general probability spaces (breaking the $\Sigma H$ paradigm). 

\item {\bf Consistency proofs} (Section \ref{sec:consistency}): We prove that the proposed estimators converge to the true value of the corresponding graph divergence measures as the number of observed samples increases in a general setting which includes several cases: (1) the data has some discrete and some continuous valued components (2) some (or all) of the components themselves are discrete-continuous mixtures (3) the data is real-valued but does not have a joint density on the entire space but is supported on a low-dimensional manifold. 

\item {\bf Numerical results} (Section \ref{sec:experiments}): Extensive numerical results demonstrate that (1) existing algorithms have severe failure modes in general probability spaces (strongly negative values, for example), and (2) our proposed estimator achieves consistency as well as significantly better finite-sample performance.
\end{enumerate}

\begin{figure}
  \begin{subfigure}[t]{0.3\textwidth}
  \centering
    \includegraphics[width=0.5\textwidth]{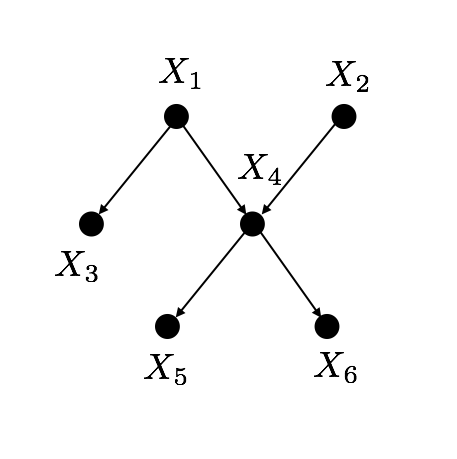}
    \caption{}
    \label{fig:bayesian-graph}
  \end{subfigure}
  \begin{subfigure}[t]{0.3\textwidth}
  \centering
    \includegraphics[width=0.5\textwidth]{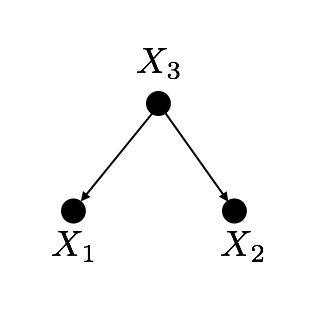}
      \caption{}
    \label{fig:markov}
  \end{subfigure}
  \begin{subfigure}[t]{0.3\textwidth}
  \centering
    \includegraphics[width=0.5\textwidth]{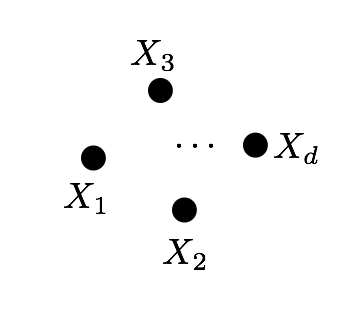}
      \caption{}
    \label{fig:independent}
  \end{subfigure}
  \caption{(a) An example of Bayesian Network $\mathcal{G}$ with $\gls{graph_dist}_{X}$ as the induced distribution $\mathbb{P}_{X_1} \mathbb{P}_{X_2}\mathbb{P}_{X_3|X_1}$ $\mathbb{P}_{X_4|{X_1, X_2}}\mathbb{P}_{X_5|X_4}\mathbb{P}_{X_6|X_4}$. (b) A Bayesian Network $\mathcal{G}$ inducing a Markov chain $\mathbb{P}_{X_3} \mathbb{P}_{X_1|X_3}\mathbb{P}_{X_2|X_3}$. (c) A Bayesian Network $\mathcal{G}$ with $\gls{graph_dist}_{X}$ as the induced distribution $\mathbb{P}_{X_1} \mathbb{P}_{X_2}\cdots \mathbb{P}_{X_d}$. }
\end{figure}

\section{Graph Divergence Measure}
\label{sec:gdm}

In this section, we define the family of graph divergence measures. To begin with, we first define some notational preliminaries. We denote any random variable by an uppercase letter such as $X$. The sample space of the variable $X$ is denoted by $\mathcal{X}$ and any value in $\mathcal{X}$ is denoted by the lowercase letter $x$. For any subset $A \subseteq \mathcal{X}$, the probability of $A$ for a given distribution function $\mathbb{P}_X(.)$ over $X$ is denoted by $\mathbb{P}_X(A)$. We note that the random variable $X$ can be a $d$-dimensional \textit{vector} of random variables, i.e. $X \equiv (X_1, \ldots, X_d)$. The $N$ observed samples drawn from the distribution $\mathbb{P}_X$ are denoted by $x^{(1)}, x^{(2)}, \ldots, x^{(N)}$, i.e. $x^{(i)}$ is the $i^{th}$ observed sample.

Sometimes we might be interested in a subset of components of a random variable, $S \subseteq \{ X_1, \ldots, X_d \}$ instead of the entire vector $X$. Accordingly, the sample space of the variable $S$ is denoted by $\mathcal{S}$. For instance, $X=(X_1, X_2, X_3, X_4)$ and $S=(X_1,X_2)$. Throughout the entire paper, unless otherwise stated, there is a one-to-one correspondence between the notations of $X$ and any $S$. For example for any value $x \in \mathcal{X}$, the  corresponding value in $\mathcal{S}$ is simply denoted by $s$. Further,  $s^{(i)} \in \mathcal{S}$ represents the lower-dimensional sample corresponding to the $i$th observed  sample $x^{(i)} \in \mathcal{X}$. Furthermore, any marginal distribution defined over $S$ with respect to $\mathbb{P}_X$ is denoted by $\mathbb{P}_S$.

Consider a directed acyclic graph (DAG) $\mathcal{G}$ defined over $d$ nodes (corresponding to the $d$ components of the random variable $X$). A probability measure $Q$ over $X$ is said to be compatible with the graph $\mathcal{G}$ if it is a Bayesian network on $\mathcal{G}$. Given a graph $\mathcal{G}$ and a distribution $\mathbb{P}_X$, there is a natural measure $\gls{graph_dist}_X(.)$ which is compatible with the graph and is defined as follows: 
\begin{equation} \gls{graph_dist}_X = \prod_{l=1}^d \mathbb{P}_{ X_l | \gls{pa_xl_var} } \end{equation}
 where $\gls{pa_xl_var} \subset X$ is the set of the \textit{parent} nodes of the random variable $X_l$, with the sample space denoted by $\gls{pa_xl_space}$, and the sample values $\gls{pa_xl_val}$ corresponding to $x$.  The distribution $\mathbb{P}_{ X_l | \gls{pa_xl_var}} $ is the \textit{conditional} distribution of $X_l$ given $\gls{pa_xl_var}$. Throughout the paper, whenever mentioning the variable $X_l$ with its own parents $\gls{pa_xl_var}$ we indicate it by $\gls{pa_n_xl_var}$, i.e. $\gls{pa_n_xl_var} \equiv \Big( X_l,\gls{pa_xl_var} \Big)$. An example is shown in Fig. \ref{fig:bayesian-graph}.

We note the fact that $ \mathbb{P}_{S | X \setminus S}$ is well defined for any subset of variables $S \subset X$. Therefore if we let $S = X \setminus \gls{pa_xl_var}$, then $\mathbb{P}_{X \setminus \gls{pa_xl_var} | \gls{pa_xl_var}}$ is well defined for any $l\in \{ 1, \ldots, d\}$. By marginalizing over $X \setminus \gls{pa_n_xl_var}$ we see that $\mathbb{P}_{ X_l | \gls{pa_xl_var} }$ and thus the distribution $\gls{graph_dist}_X$ is well defined. 

The graph divergence measure is now defined as a function of the probability measure $\mathbb{P}_X$ and the graph $\mathcal{G}$.  In this work we will focus only on the KL Divergence as being the distance metric, hence unless otherwise stated $D(\cdot\parallel\cdot)=D_{KL}(\cdot\parallel\cdot)$. Let's first consider the case where $\mathbb{P}_X$ is absolutely continuous with respect to $\gls{graph_dist}_X$ and hence the \textit{Radon-Nikodym} derivative  $d\mathbb{P}_X/d\gls{graph_dist}_X$ exists. Therefore for a given set of random variables $X$ and a Bayesian Network $\mathcal{G}$, we define \textbf{Graph Divergence Measure ($\gls{graph_distance_measure}$)} as : 
\begin{equation} 
\gls{graph_distance_measure}(X,\mathcal{G}) = D(\mathbb{P}_X \| \gls{graph_dist}_X) = \int_{\mathcal{X}} \log \frac{d\mathbb{P}_X}{d\gls{graph_dist}_X} d\mathbb{P}_X 
\label{eq:measure}\end{equation}
Here we implicitly assume that $\log \big( d\mathbb{P}_X / d\gls{graph_dist}_X \big)$ is measurable and integrable with respect to the measure $\mathbb{P}_X$. The $\gls{graph_distance_measure}$ is set to infinity wherever Radon-Nikodym derivative does not exist. It is clear that $\gls{graph_distance_measure}(X,\mathcal{G})=0$ if and only if the data distribution is compatible with the graphical model, thus the  $\gls{graph_distance_measure}$ can be thought of as a {\em measure of incompatibility with the given graphical model structure}. 


We now have relevant variational characterization as below on our graph divergence measure, which can be harnessed to compute upper and lower bounds (More details in \textbf{supplementary material}):
\begin{proposition}
For a random variable $X$, a DAG $\mathcal{G}$, let $\Pi(\mathcal{G})$ be the set of measures $\mathbb{Q}_X$ defined on the Bayesian Network $\mathcal{G}$, then $\gls{graph_distance_measure}(X,\mathcal{G})=\inf_{\mathbb{Q}_X\in\Pi(\mathcal{G})}D(\mathbb{P}_X \| \mathbb{Q}_X) $. 

Furthermore, let $\mathcal{C}$ denote the set of functions $h:\mathcal{X}\rightarrow \mathbb{R}$ such that $\mathbb{E}_{\mathbb{Q}_X}\left[\exp(h(X))\right]<\infty$. If $\mathbb{GDM}(X,\mathcal{G})<\infty$, then for every $h\in\mathcal{C}$, $\mathbb{E}_{\mathbb{P}_X}\left[h(X)\right]$ exists and:
\begin{align}
\mathbb{GDM}(X,\mathcal{G})=\sup_{h\in\mathcal{C}} \mathbb{E}_{\mathbb{P}_X}\left[h(X)\right] - \log \mathbb{E}_{\mathbb{Q}_X}\left[\exp(h(X))\right].
\end{align}
\end{proposition}

\vspace{-.15in}
\subsection{Special cases}
\label{sec:math-inf-quantities}

For specific choices of $X$ and Bayesian Network, $\mathcal{G}$, the Equation \ref{eq:measure} is reduced to the well-known information measures. Some examples of these measures are as follows:\\
\textbf{Mutual Information (MI):} $X=(X_1, X_2)$ and $\mathcal{G}$ has no directed edge between $X_1$ and $X_2$. Thus $\gls{graph_dist}_X = \mathbb{P}_{X_1}.\mathbb{P}_{X_2}$, and we get, $\gls{graph_distance_measure}(X,\mathcal{G}) = I(X_1;X_2)  = D(\mathbb{P}_{X_1X_2} \| \mathbb{P}_{X_1}\mathbb{P}_{X_2})$. 

\textbf{Conditional Mutual Information (CMI):} We recover the conditional mutual information of $X_1$ and $X_2$ given $X_3$ by constraining $\mathcal{G}$ to be the one in Fig. \ref{fig:markov}, since $\gls{graph_dist}_X = \mathbb{P}_{X_3}.\mathbb{P}_{X_2|X_3}.\mathbb{P}_{X_1|X_3}$, i.e., $\gls{graph_distance_measure}(X,\mathcal{G}) = I(X_1;X_2|X_3)  = D(\mathbb{P}_{X_1X_2X_3} \| \mathbb{P}_{X_1|X_3}\mathbb{P}_{X_2|X_3}\mathbb{P}_{X_3})$.\\

\textbf{Total Correlation (TC):} When $X=(X_1, \cdots, X_d)$, and $\mathcal{G}$ is the graph with no edges (as in Fig. \ref{fig:independent}, we recover the total correlation of the random variables $X_1,\ldots,X_d$ since $\gls{graph_dist}_X = \mathbb{P}_{X_1}\ldots \mathbb{P}_{X_d}$, i.e., $\gls{graph_distance_measure}(X,\mathcal{G}_{dc}) =TC(X_1,\ldots,X_d)  = D(\mathbb{P}_{X_1 \ldots X_d} \| \mathbb{P}_{X_1}\ldots \mathbb{P}_{X_d})$
\\

\textbf{Multivariate Mutual Information (MMI) :} While the MMI defined by \cite{mcgill1954multivariate} is not positive in general,there is a different definition by \cite{chan2015multivariate} which is both non-negative and has an operational interpretation. Since MMI can be defined as the optimal total correlation after clustering, we can utilize our definition to define MMI (cf. \textbf{supplementary material}).



\textbf{Directed Information :} Suppose there are two stationary random processes $X$ and $Y$, the directed information rate from $X$ to $Y$ as first introduced by Massey \cite{massey1990directedinformation} is defined as:
$$ I( X \rightarrow Y ) = \frac{1}{T}\sum_{t=1}^T I \left( X^t ; Y_t \middle| Y^{t-1} \right)$$ \\
It can be seen that the directed information can be written as:
$$ I( X \rightarrow Y ) = \gls{graph_distance_measure}\Big( (X^T,Y^T), \mathcal{G}_I \Big) - \gls{graph_distance_measure}\Big( (X^T,Y^T), \mathcal{G}_C \Big) $$ \\
where the graphical model $\mathcal{G}_I$ correponds to the \textit{independent} distribution between $X^T$ and $Y^T$, and $\mathcal{G}_C$ corresponds to the \textit{causal} distribution from $X$ to $Y$ (more details provided in \textbf{supplementary material}).



\section{Estimators}
\label{sec:estimators}

\subsection{Prior Art}

Estimators for entropy date back to Shannon, who guesstimated the entropy rate of English \cite{shannon1951prediction}. Discrete entropy estimation is a well-studied topic and minimax rate of this problem is well-understood as a function of the alphabet size \cite{paninski2003estimation,jiao2015minimax,wu2016minimax}. The estimation of differential entropy is considerably  harder and also studied extensively in literature \cite{beirlant1997nonparametric, nemenman2002entropy, miller2003new, lee2010sample, lesniewicz2014expected, sricharan2012estimation, singh2014exponential} and can be broadly divided into two groups; based on either Kernel density estimates  \cite{kandasamy2015nonparametric,gao2016breaking} or based on k-nearest-neighbor estimation \cite{sricharan2013ensemble,jiao2017nearest, pal2010estimation}. In a remarkable work, Kozachenko and Leonenko suggested a nearest neighbor method for entropy estimation \cite{kozachenko1987sample} which was then generalized to a $k$th nearest neighbor approach \cite{singh2003nearest}. In this method, the distance to the $k$th nearest neighbor (KNN) is measured for each data-point, and based on this the probability density around each data point is estimated and substituted into the entropy expression. When $k$ is fixed, each density estimate is noisy and the estimator accrues a bias and a remarkable result is that the bias is distribution-independent and can be subtracted out \cite{singh2016finite}. 

While the entropy estimation problem is well-studied, mutual information and its generalizations are typically estimated using a sum of signed entropy ($H$) terms, which are estimated first; we term such estimators as $\Sigma H$ estimators. In the discrete alphabet case, this principle has been shown to be worst-case optimal \cite{han2015adaptive}. In the case of distributions with a joint density, an estimator that breaks the $\Sigma H$ principle is the KSG estimator \cite{kraskov2004estimating}, which builds on the KNN estimation paradigm but couples the estimates in order to reduce the bias. This estimator is widely used and has excellent practical performance. The original paper did not have any consistency guarantees and its convergence rates were recently established \cite{pramod_demystify_2018}.  There have been some extensions to the KSG estimator for other information measures such as conditional mutual information \cite{runge2017conditional, frenzel2007partial}, directed information \cite{vejmelka2008inferring} but none of them show theoretical guarantees on consistency of the estimators, furthermore they fail completely in mixture distributions.

When the data distribution is neither discrete nor admits a joint density, the $\Sigma H$ approach is no longer feasible as the individual entropy terms are not well defined. This is the regime of interest in our paper. Recently, Gao et al  \cite{gao2017mixture} proposed a mutual-information estimator based on KNN principle, which can handle such continuous-discrete mixture cases, and the consistency was demonstrated. However it is not clear how it generalizes to even Conditional Mutual Information (CMI) estimation, let alone other generalizations of mutual information. In this paper, we build on that estimator in order to design an estimator for general graph divergence measures and establish its consistency for generic probability spaces. 

\subsection{Proposed Estimator}

The proposed estimator is given in Algorithm \ref{algm:gdm} where $\psi(\cdot)$ is the digamma function and $\mathbf{1}_{\{\cdot\}}$ is the indicator function. The process is schematically shown in Fig. \ref{fig:estimator_schematic} (cf. \textbf{supplementary material}). We used the $\ell_\infty$-norm everywhere in our algorithm and proofs. 
\begin{algorithm}[htbp]
\caption{Estimating \textbf{Graph Divergence Measure} $\gls{graph_distance_measure}(X, \mathcal{G})$}
\label{algm:gdm}.

\begin{algorithmic}[1]
\Require{\textbf{Parameter:} $k\in\mathbf{Z}^+$, \textbf{Samples:} $x^{(1)}, x^{(2)}, \ldots, x^{(N)}$, \textbf{Bayesian Network:} $\mathcal{G}$ on \textbf{Variables:} ${\cal X}=(X_1,X_2,\cdots,X_d)$}
\Ensure{$\widehat{\gls{graph_distance_measure}}^{(N)}(X,\mathcal{G})$}

 \For{\texttt{$i=1$ to $N$}}
 \State \textbf{Query:} 
 \State $\rho_{k,i}$ = $\ell_\infty$-distance to the $k$th nearest neighbor of $x^{(i)}$ in the space $\mathcal{X}$
\State \textbf{Inquire:}
\State  $\tilde{k}_i$ = \# points within the $\rho_{k,i}$-neighborhood of $x^{(i)}$ in the space $\mathcal{X}$
\State $n_{\gls{pa_xl_var}}^{(i)}$ = \# points within the $\rho_{k,i}$-neighborhood of $x^{(i)}$ in the space $\gls{pa_xl_space}$
\State $n_{\gls{pa_n_xl_var}}^{(i)}$ = \# points within the $\rho_{k,i}$-neighborhood of $x^{(i)}$ in the space $\gls{pa_n_xl_space}$
\State \textbf{Compute:}
\State $\zeta_i = \psi(\tilde{k}_i) + \sum_{l=1}^d \left(  \mathbf{1}_{\{  \gls{pa_xl_var} \neq \emptyset \}} \log \left( n_{\gls{pa_xl_var}}^{(i)} +1 \right) - \log \left( n_{\gls{pa_n_xl_var}}^{(i)} +1 \right) \right)$
\EndFor
\State Final Estimator: 
\begin{equation}\label{eq:estimator}
\widehat{\gls{graph_distance_measure}}^{(N)}(X,\mathcal{G}) = \frac{1}{N} \sum_{i=1}^N \zeta_i + \left( \sum_{l=1}^d \mathbf{1}_{\{ \gls{pa_xl_var} =\emptyset \}}-1 \right) \log N
\end{equation}
\end{algorithmic}
\end{algorithm}

%

The estimator intuitively estimates the $\gls{graph_distance_measure}$ by the \textit{resubstitution estimate} $\frac{1}{N} \sum_{i=1}^N \log \hat{f}(x^{(i)})$ in which $\hat{f}(x^{(i)})$ is the estimation of Radon-Nikodym derivative at each sample $x^{(i)}$. If $x^{(i)}$ lies in a region where there is a density, the RN derivative is equal to $g_X(x^{(i)})/\bar{g}_X(x^{(i)})$ in which $g_X(.)$ and $\bar{g}_X(.)$ are density functions corresponding to $\mathbb{P}_X$ and $\gls{graph_dist}_X$ respectively. On the other hand, if  $x^{(i)}$ lies on a point where there is a discrete mass, the RN derivative will be equal to $h_X(x^{(i)})/\bar{h}_X(x^{(i)})$ in which $h_X(.)$ and $\bar{h}_X(.)$ are mass functions corresponding to $\mathbb{P}_X$ and $\gls{graph_dist}_X$ respectively.

The density function $\bar{g}_X(x^{(i)})$ can be written as $\prod_{l=1}^d \left( g_{\gls{pa_n_xl_var}}(\gls{pa_n_xl_val}^{(i)}) / g_{\gls{pa_xl_var}}(\gls{pa_xl_val}^{(i)}) \right)$ for continuous components. Equivalently, the mass function $\bar{h}_X(x^{(i)})$ can be written as $\prod_{l=1}^d \left( h_{\gls{pa_n_xl_var}}(\gls{pa_n_xl_val}^{(i)}) / h_{\gls{pa_xl_var}}(\gls{pa_xl_val}^{(i)}) \right)$. Thus we need to estimate the density functions $g(.)$ and the mass functions $h(.)$ according to the type of $x^{(i)}$. The existing $k$th nearest neighbor algorithms will suffer while estimating the mass functions $h(.)$, since $\rho_{n_S,i}$  (the distance to the $n_S$-th nearest neighbor in subspace $\mathcal{S}$) for such points will be equal to zero for large $N$. Our algorithm, however, is designed in a way that it's capable of approximating both $g(.)$ functions as $\approx \frac{n_S}{N} \frac{1}{(\rho_{n_S,i})^{d_S}}$ and $h(.)$ functions as $\approx \frac{n_S}{N}$ dynamically for any subset $S \subseteq X$. This is achieved by setting $\rho_{n_S,i}$ terms such that all of them cancel out, yielding the estimator as in Eq. (\ref{eq:estimator}).

\section{Proof of Consistency}\label{sec:consistency}

The proof of consistency for our estimator consists of two steps: First we prove that the expected value of the estimator in Eq. (\ref{eq:estimator}) converges to the true value as $N \rightarrow \infty$ , and second we prove that the variance of the estimator converges to zero as $N \rightarrow \infty$. 
Let's begin with the definition of $P_X(x,r)$:
\begin{equation} P_X(x,r) = \mathbb{P}_X \big\{ a \in \mathcal{X} : \| a- x \|_{\infty} \leq r \big\} = \mathbb{P}_X \Big\{ B_r(x) \Big\} \label{eq:P_xr}\end{equation}
Thus $P_X(x,r)$ is the probability of a hypercube with the edge length of $2r$ centered at the point $x$. We then state the following assmuptions:

\begin{assumption}\label{assumptions} 
We make the following assumptions to prove the consistency of our method:

\begin{enumerate}
\item $k$ is set such that $\lim_{N \rightarrow \infty} k = \infty$ and $\lim_{N \rightarrow \infty} \frac{k \log N}{N}=0$. \label{assumption_kn}
\item The set of discrete points $\{ x : P_X (x, 0) > 0 \}$ is finite. \label{assumption_finitediscrete}
\item $\int_{\mathcal{X}} \big| \log f(x) \big| d\mathbb{P}_X < +\infty$, where $f \equiv d\mathbb{P}_X/d\gls{graph_dist}_X$ is Radon-Nikodym derivative. \label{assumption_logintegrable}
\end{enumerate}

\end{assumption} 

The Assumption \ref{assumptions}.\ref{assumption_kn} with \ref{assumptions}.\ref{assumption_finitediscrete} controls the boundary effect between the continuous and the discrete regions; with this assumption we make sure that all the $k$ nearest neighbors of each point belong to the same region almost surely (i.e. all of them are either continuous or discrete).  
Assumption \ref{assumptions}.\ref{assumption_logintegrable} is the log-integrability of the Radon-Nikodym derivative. These assumptions are satisfied under mild technical conditions whenever the distribution $\mathbb{P}_X$ over the set $\mathcal{X}$ is (1) finitely discrete; (2) continuous; (3) finitely discrete over some dimensions and continuous over some others; (4) a mixture of the previous cases; (5) has a joint density supported over a lower dimensional manifold. These cases represent almost all the real world data.

As an example of a case not conforming to the aforementioned cases, we can consider singular distributions, among which the \textit{Cantor distribution} is a significant example whose cumulative distribution function is the Cantor function. This distribution has neither a probability density function nor a probability mass function, although its cumulative distribution function is a continuous function. It is thus neither a discrete nor an absolutely continuous probability distribution, nor is it a mixture of these.

The Theorem \ref{thm:mean_convergence} formally states the mean-convergence of the estimator while Theorem \ref{thm:variance_convergence} formally states that convergence of the variance to zero.
\begin{theorem}\label{thm:mean_convergence}
Under the Assumptions \ref{assumptions}, we have $\lim_{N \rightarrow \infty} \mathbb{E} \left[ \widehat{\gls{graph_distance_measure}}^{(N)}(X,\mathcal{G}) \right] = \gls{graph_distance_measure}(X,\mathcal{G})$.
\end{theorem}

\begin{theorem}\label{thm:variance_convergence}
In addition to the Assumptions \ref{assumptions}, assume that we have $( k_N \log N )^2 / N \rightarrow 0 $ as $N$ goes to infinity. Then we have $\lim_{N \rightarrow \infty} \text{Var} \left[ \widehat{\gls{graph_distance_measure}}^{(N)}(X,\mathcal{G}) \right] = 0$.
\end{theorem}

The Theorems \ref{thm:mean_convergence} and \ref{thm:variance_convergence} combined yield the consistency of the estimator \ref{eq:estimator}. 

The proof of the Theorem \ref{thm:mean_convergence} starts with writing the Radon-Nikodym derivative explicitly. Then we need to upper-bound the term 
$\big| \mathbb{E} \big[ \widehat{\gls{graph_distance_measure}}^{(N)}(X,\mathcal{G}) \big] - \gls{graph_distance_measure}(X,\mathcal{G}) \big|$. To achieve this goal, we segregate the domain of $\mathcal{X}$ into three parts as $\mathcal{X} = \Omega_1 \cup \Omega_2 \cup \Omega_3 $ where  $\Omega_1 = \{x: f(x)=0 \}$,  $\Omega_2 = \{x: f(x)>0, P_X(x,0)>0 \}$ and $\Omega_3 = \{x: f(x)>0, P_X(x,0)=0 \}$.
We will show that $\mathbb{P}_X(\Omega_1)=0$. The sets $\Omega_2$ and $\Omega_3$ correspond to the discrete and continuous regions respectively. Then for each of the two regions, we introduce an upperbound which goes to zero as $N$ grows boundlessly. Thus equivalently we show the mean of the estimate $\zeta_1$ is close to $\log f(x)$ for any $x$.

The proof of the Theorem \ref{thm:variance_convergence} is based on the Efron-Stein inequality, which upperbounds any estimator for any quantity from the observed samples $x^{(1)}, \ldots, x^{(N)}$. For any sample $x^{(i)}$, we then upperbound the term $\left| \zeta_i(X) - \zeta_i(X_{\setminus j}) \right|$ by segregating the samples into various cases, and examining each case separately. $\zeta_i(X)$ is the estimate using all the samples $x^{(1)}, \ldots, x^{(N)}$ and $\zeta_i(X_{\setminus j})$ is the estimate when the $j$th sample is removed. Summing up over all the $i$'s, we obtain an upper-bound which will converge to $0$ as $N$ goes to infinity.

\section{Empirical Results}\label{sec:experiments}

In this section, we evaluate the performance of our proposed estimator in comparison with other estimators via numerical experiments. The estimators evaluated here are our estimator referred to as \textit{GDM}, the plain KSG-based estimators for continuous distributions to which we refer by \textit{KSG}, the \textit{binning} estimators and the noise-induced $\Sigma H$ estimators. A more detailed discussion can be found in Section~\ref{sec:exp_details}.

\textbf{Experiment 1: Markov chain model with continuous-discrete mixture.}
For the first experiment, we simulated an $X$-$Z$-$Y$ Markov chain model in which the random variable $X$ is a uniform random variable $\mathcal{U}(0,1)$ clipped at a threshold $0 <\alpha_1 < 1$ from above. Then $Z = \min \left( X, \alpha_2 \right)$ and $Y = \min \left( Z, \alpha_3 \right)$ in which $ 0 < \alpha_3 <\alpha_2 <\alpha_1$.
We simulated this system for various numbers of samples, setting $\alpha_1=0.9$, $\alpha_2=0.8$ and $\alpha_3=0.7$. For each set of samples we estimated $I(X;Y|Z)$ via different methods. The theory value for $I(X;Y|Z)$ is $0$. The results are shown in Figure \ref{fig:exp1_cmi_vs_n}. We can see that in this regime, only the GDM estimator can correctly converge. The KSG estimator and the $\Sigma H$ estimator show high negative biases and the binning estimator shows a positive bias.

\textbf{Experiment 2: Mixture of AWGN and BSC channels with variable error probability.}
For the second scheme of our experiments, we considered an Additive White Gaussian Noise (AWGN) Channel in parallel with a Binary Symmetric Channel (BSC) where only one of the two can be activated at a time. The random variable $Z = \min ( \alpha, \tilde{Z} )$ where $\tilde{Z} \sim U(0,1)$ controls which channel is activated; i.e. if $Z$ is lower than the threshold $\beta$, activate the AWGN channel, otherwise initiate the BSC channel where $Z$ also determines the error probability at each time point. We set $\alpha=0.3$, $\beta=0.2$, BSC channel input as $X \sim \text{Bern}(0.5)$, and AWGN input and noise deviation as $\sigma_X=1$ and $\sigma_N=0.1$ respectively, and obtained the estimates of $I(X;Y|Z,Z^2,Z^3)$ for various estimators. While the theory value is equal to $I(X;Y|Z)=0.53241$, yet it's conditioned over a low-dimensional manifold in a high-dimensional space. The results are shown in Figure \ref{fig:exp2_cmi_ld_vs_n}. Similar to the previous experiment, the GDM estimator can correctly converge to the true value. The $\Sigma H$ and binning estimators show a negative bias, and the KSG estimator gets totally lost.

\textbf{Experiment 3: Total Correlation for independent mixtures.}
In this experiment, we estimate the total correlation of three independent variables $X$, $Y$ and $Z$. The samples for the variable $X$ are generated in the following fashion:
First toss a fair coin, if heads appears we fix $X$ at $\alpha_X$, otherwise we draw $X$ from a uniform distribution between $0$ and $1$. samples from $Y$ and $Z$ are also generated in the same way independently with parameters $\alpha_Y$ and $\alpha_Z$ respectively. For this setup, $TC(X,Y,Z)=0$. We set $\alpha_X=1$, $\alpha_Y=1/2$ and $\alpha_Z=1/4$, and generated various datasets with different lengths. Then estimated total correlation values are shown in the Figure \ref{fig:exp3_tc_vs_n}.  

\textbf{Experiment 4: Total Correlation for independent uniforms with correlated zero-inflation.}
Here we first consider four auxiliary uniform variables $\tilde{X}_1$, $\tilde{X}_2$, $\tilde{X}_3$ and $\tilde{X}_4$ which are taken from $\mathcal{U}(0.5,1.5)$. Then each sample is erased with a Bernoulli probability; i.e. $X_1=\alpha_1 \tilde{X}_1$, $X_2=\alpha_1 \tilde{X}_2$ and $X_3=\alpha_2 \tilde{X}_3$, $X_4=\alpha_2 \tilde{X}_4$   in which $\alpha_1 \sim \text{Bern}(p_1)$ and $\alpha_2 \sim \text{Bern}(p_2)$. As we see, after zero-inflation $X_1$ and $X_2$ become correlated, so do $X_3$ and $X_4$ while still $(X_1,X_2) \indep (X_3,X_4)$.
 In the experiment, we set $p_1=p_2=0.6$. The results of running different algorithms over the data can be seen in Figure \ref{fig:exp4_tc_vs_n}. For the total correlation experiments 3 and 4, similar to that of conditional mutual information in experiments 1 and 2, only the GDM estimator can best estimate the true value. The estimator $\Sigma H$ was removed from the figures due to its high inaccuracy.

\begin{figure}
\centering
	\begin{subfigure}[t]{.45\textwidth}
 	\centering
	\includegraphics[width=\textwidth,trim={.3cm 0 1cm 0},clip]{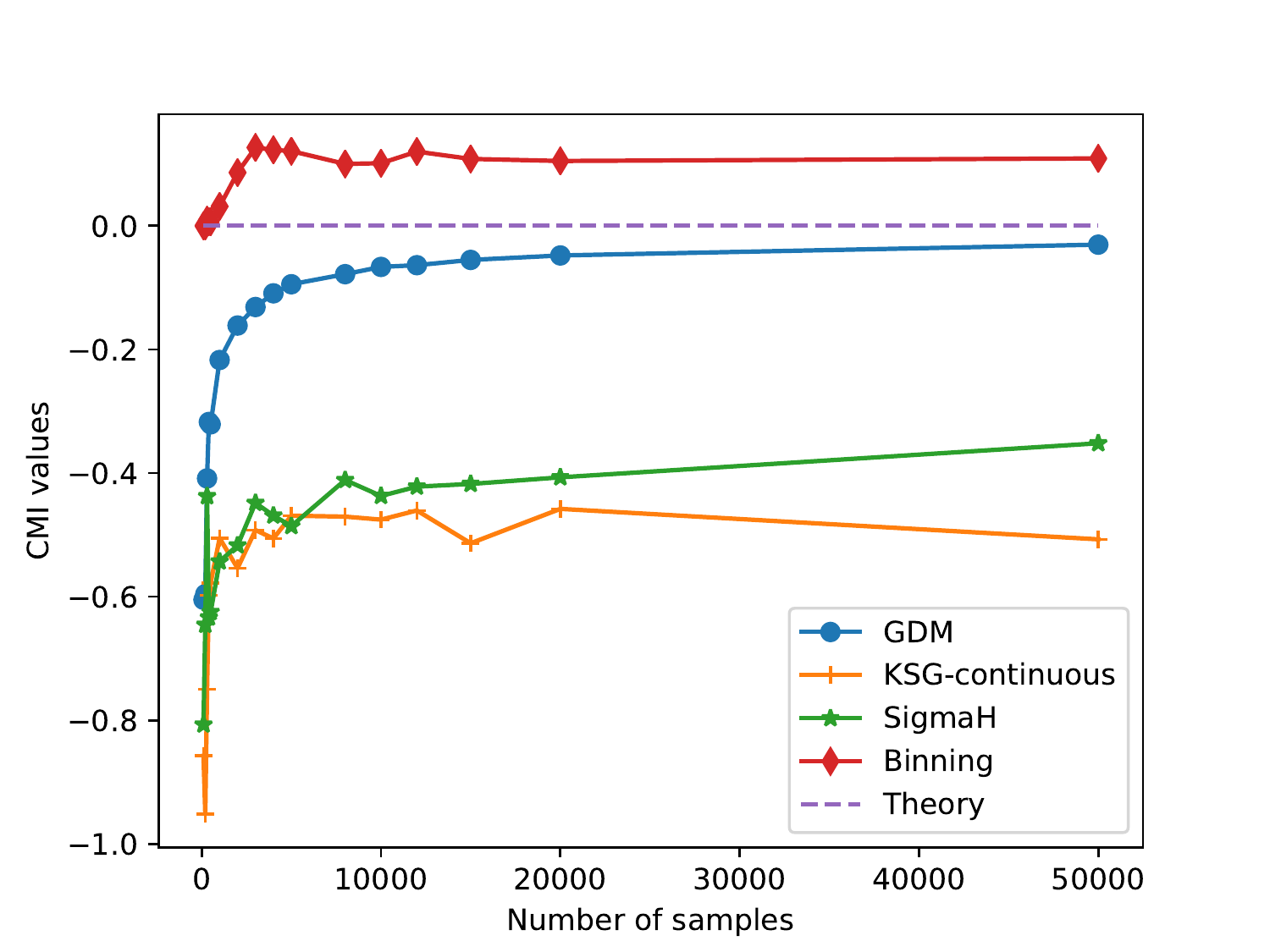}
	\caption{}
	\label{fig:exp1_cmi_vs_n}
	\end{subfigure}
	\begin{subfigure}[t]{.45\textwidth}
	\centering
	\includegraphics[width=\textwidth,trim={.3cm 0 1cm 0},clip]{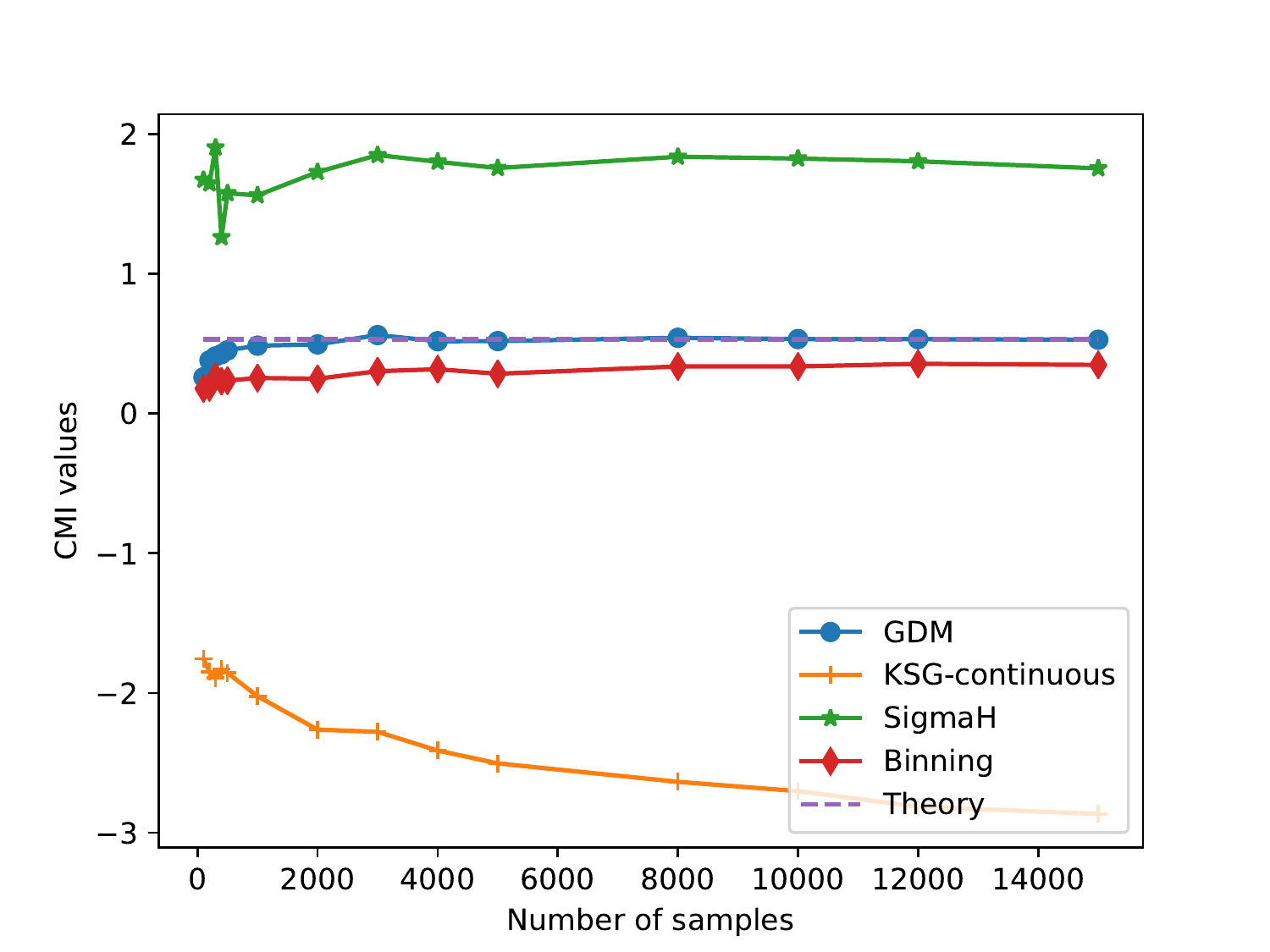}
	\caption{}
	\label{fig:exp2_cmi_ld_vs_n}
	\end{subfigure}

	\begin{subfigure}[t]{.45\textwidth}
	\centering
	\includegraphics[width=\textwidth,trim={.3cm 0 1cm 0},clip]{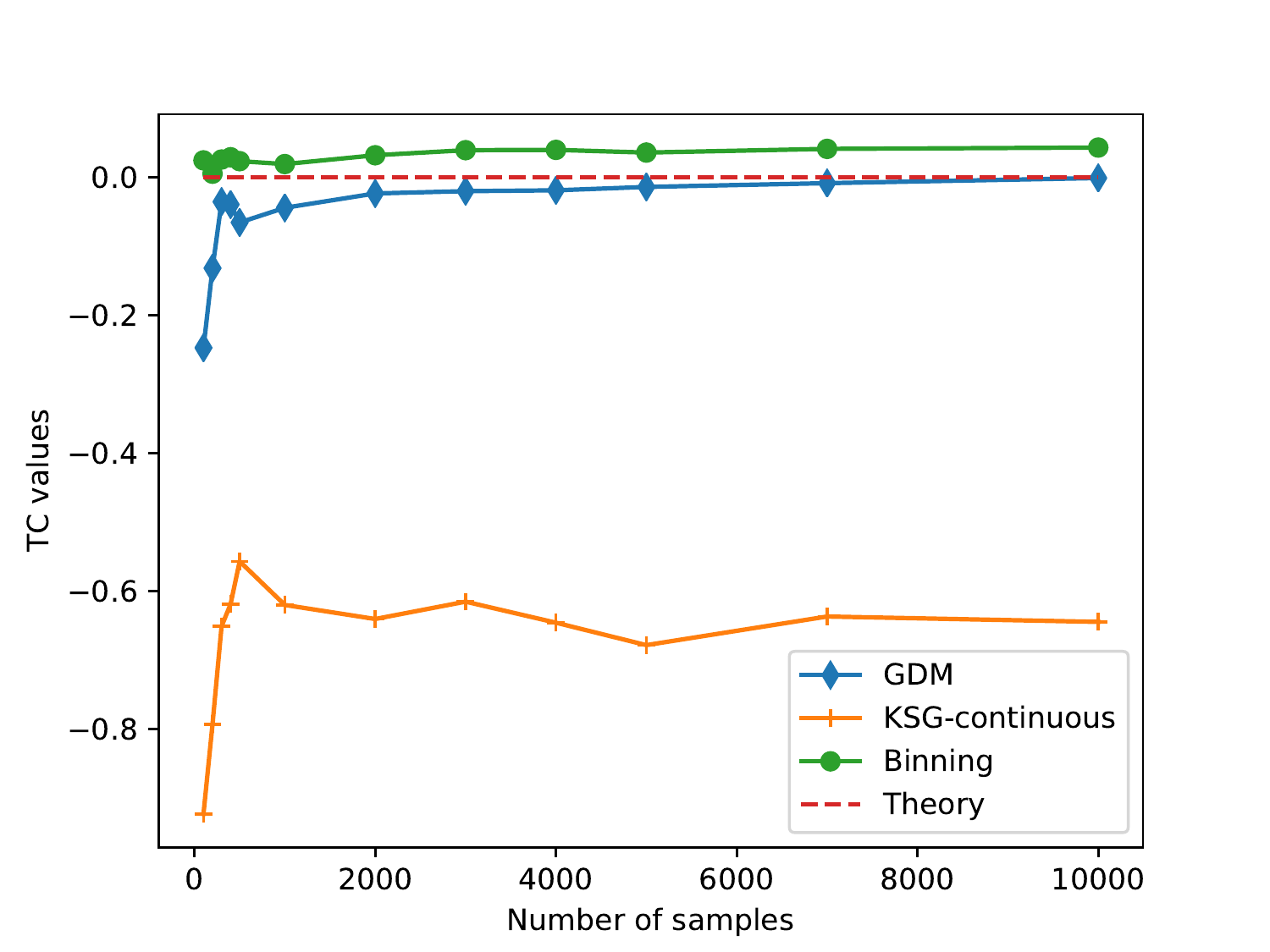}
	\caption{}
	\label{fig:exp3_tc_vs_n}
	\end{subfigure}
	\begin{subfigure}[t]{.45\textwidth}
	\centering
	\includegraphics[width=\textwidth,trim={.3cm 0 1cm 0},clip]{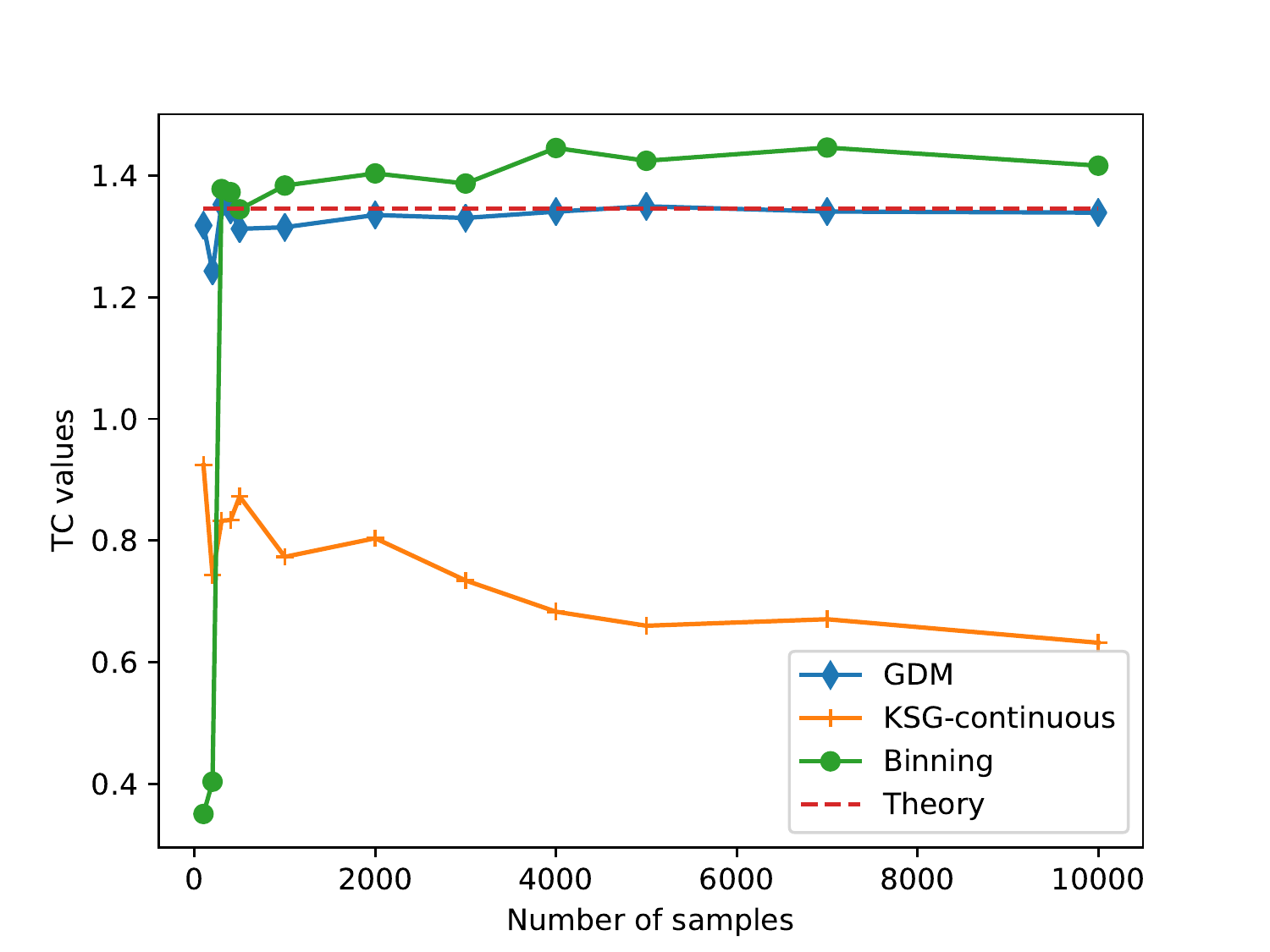}
	\caption{}
	\label{fig:exp4_tc_vs_n}
	\end{subfigure}

	\begin{subfigure}[t]{.45\textwidth}
	\centering
	\includegraphics[width=\textwidth,trim={.3cm 0 1cm 0},clip]{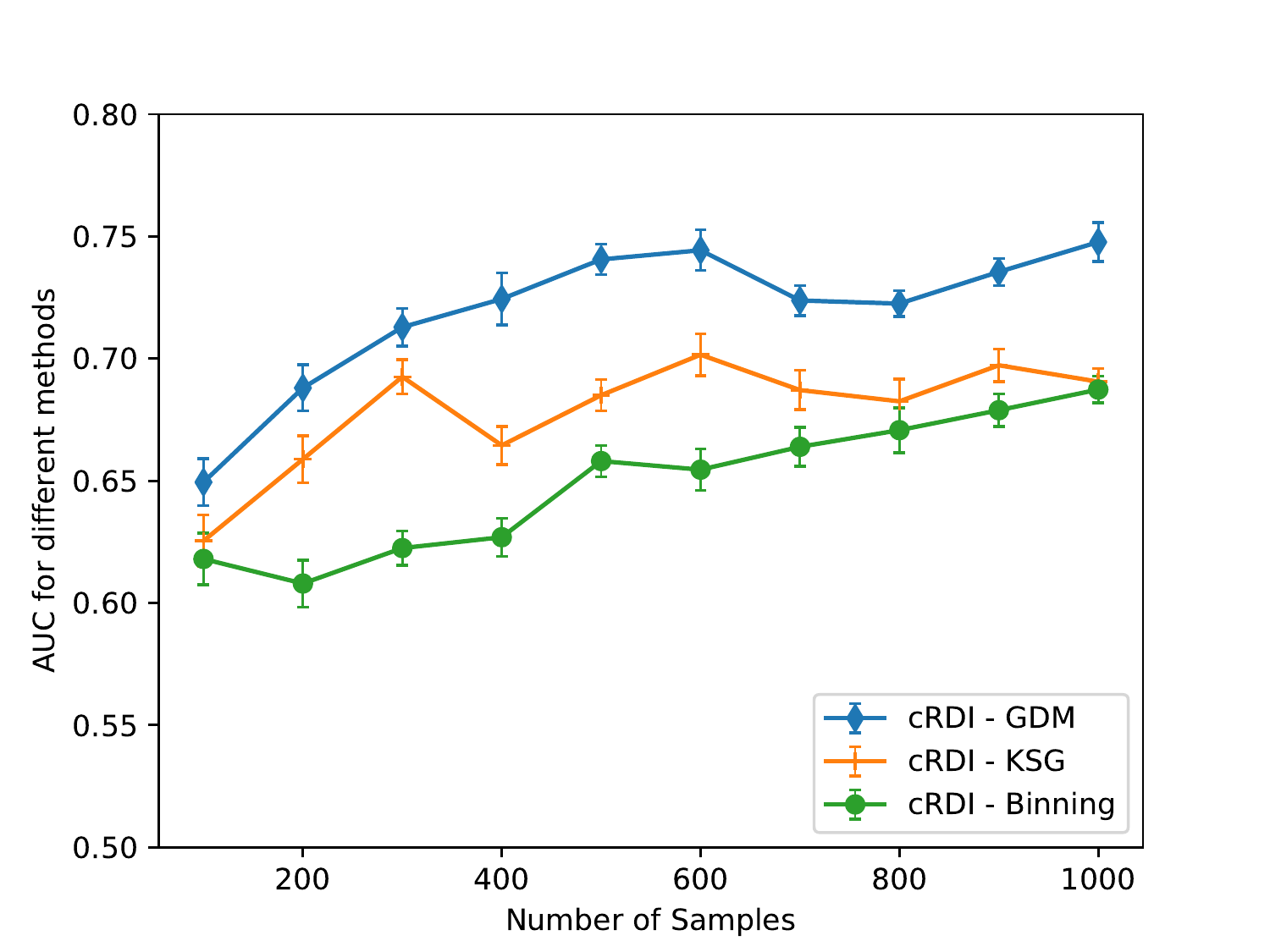}
	\caption{}
	\label{fig:grn_auc_vs_n}
	\end{subfigure}
	\begin{subfigure}[t]{.45\textwidth}
	\centering
	\includegraphics[width=\textwidth,trim={.3cm 0 1cm 0},clip]{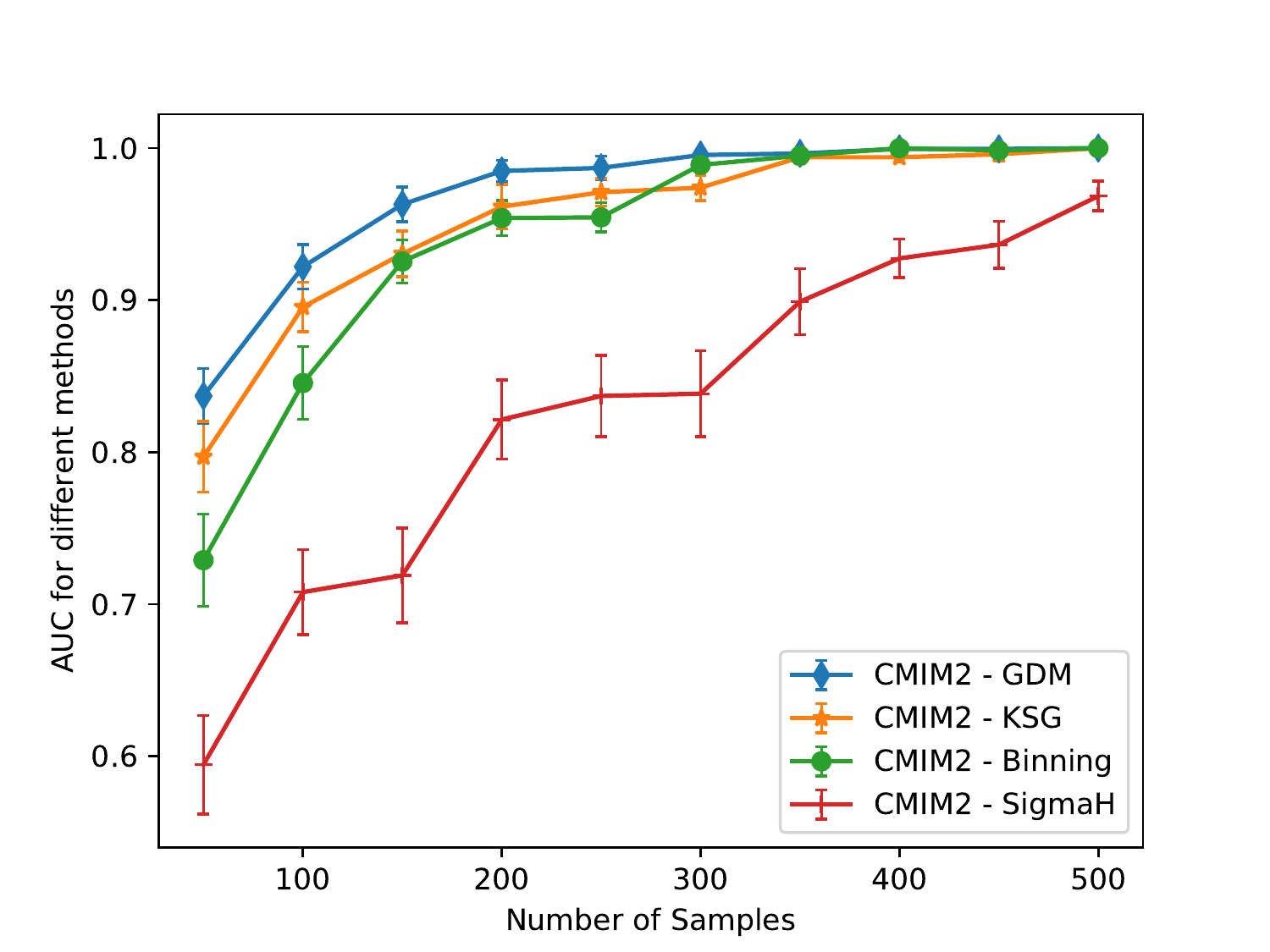}
	\caption{}
	\label{fig:featureselection_auc_vs_n}
	\end{subfigure}

\caption{The results for the experiments versus the number of samples: \ref{fig:exp1_cmi_vs_n}: The estimated CMI for the X-Z-Y Markov chain. \ref{fig:exp2_cmi_ld_vs_n}: CMI for the AWGN+BSC channels with low-dimensional $Z$ manifold. \ref{fig:exp3_tc_vs_n}: The estimated TC values for three independent variables. \ref{fig:exp4_tc_vs_n}: The estimated TC for zero-inflated variables. \ref{fig:grn_auc_vs_n}: The AUROC values for gene regulatory network inference. The error bars show the standard deviation scaled down by $0.2$. \ref{fig:featureselection_auc_vs_n}: The AUROC values for feature selection accuracy. The error bars show the standard deviations scaled down by $0.2$.}
\end{figure}

\textbf{Experiment 5: Gene Regulatory Networks.}
In this experiment we use different estimators to do Gene Regulatory Network inference based on the conditional Restricted Directed Information (cRDI) \cite{rahimzamani2016network}. We do our test on the simulated neuron cells' development process, based on a model from \cite{qiu2012understanding}. In this model, the time series vector $X$ consists of $13$ random variables each of which corresponding to a single gene in the development process. We simulated the development process for various lengths of time-series in which the noise $N \sim \mathcal{N}(0,.03)$ is added for all the genes, and every single sample is then subject to erasure (i.e. be replaced by 0s) with a probability of $0.5$. Then we applied the cRDI method utilizing various CMI estimators and then calculated the Area-Under-ROC curve (AUROC). The results are shown in Figure \ref{fig:grn_auc_vs_n}. It's seen that the cRDI method implemented with the GDM estimator outperform the other estimators by at least $\%10$ in terms of AUROC. In the tests, cRDI for each $(X_i,X_j)$ is conditioned over the node $k \neq i$ with the highest RDI value to $j$. We notice that the causal signals are highly destroyed due to the zero-inflation, so we won't expect high performance of the causal inference over the data. We did not include the $\Sigma H$ estimator results due to its very low performance.

\textbf{Experiment 6: Feature Selection by Conditional Mutual Information Maximization.}
Feature selection is an important pre-processing step in many learning tasks. The application of information theoretic measures in feature selection is well studied in the literature \cite{vergara2014review}. Among the well-known methods is the conditional mutual information maximization (CMIM) first introduced by Flueret \cite{fleuret2004fast}, a variation of which was later introduced called CMIM-2 \cite{vergara2010cmim2}. Both methods use conditional mutual information as their core measure to select the features. Hence the performance of the estimators can significantly influence the performance of the methods. In our experiment, we generated a vector $X=(X_1, \ldots, X_{15})$ of 15 random variables in which all the random variables are taken from $\mathcal{N}(0,1)$ and then each random variable $X_i$ is clipped from above at $\alpha_i$ which is initially taken randomly from $\mathcal{U}(0.25,0.3)$ and then kept constant during the sample generation. Then $Y$ is generated as $Y=\cos \big( \sum_{i=1}^5 X_i \big)$. Then we did the CMIM-2 algorithm with various CMI estimators to evaluate the performance of the estimators in extracting the relevant features $X_1, \ldots, X_5$. The AUROC values for each algorithm versus the number of samples generated are shown in Figure \ref{fig:featureselection_auc_vs_n}. The feature selection methods implemented with the GDM estimator outperform the other estimators.

\section{Discussion and Future Work}
A general paradigm of graph divergence measures and novel estimators thereof, for general probability spaces are proposed, which estimate several generalizations of mutual information. In future, we would like to derive further efficient estimators for high dimensional data. In the current work, estimators are shown to be consistent with infinite scaling of parameter $k$. In future, we would like to understand the finite $k$ performance of the estimators as well as guarantees on sample complexity and rates of convergence. Another potential direction to follow is to study the variational characterization of the graph divergence measure to design estimators. Improving the computational efficiency of the estimator is another direction of future work. Recent literature including \cite{noshad2018scalable} provide a new methodology to estimate mutual information in a computationally efficient manner and leveraging these ideas for the generalized measures and general proabability distributions can be a promising direction ahead.

\section{Acknowledgement}
This work was partially supported by NSF grants 1651236, 1703403 and NIH grant 5R01HG008164.

\newpage\bibliography{ref}
\bibliographystyle{ieeetr}


\newpage\appendix



\section{Pictorial Representation of Estimators}
The Figure \ref{fig:step_1} represents the step 1 (distance query or \textbf{Query} step in the algorithm) and the Figure \ref{fig:step_2} represents the steps 2 and 3 (numbers inquiry, or the \textbf{Inquire} step in the algorithm). Note that in the graphics we used $\ell_2$-norm to give better intuition on the process, while in our proofs and simulations, we use $\ell_{\infty}$-norm.

\begin{figure}
\begin{subfigure}[t]{.5\textwidth}
\centering
	\includegraphics[scale=.5]{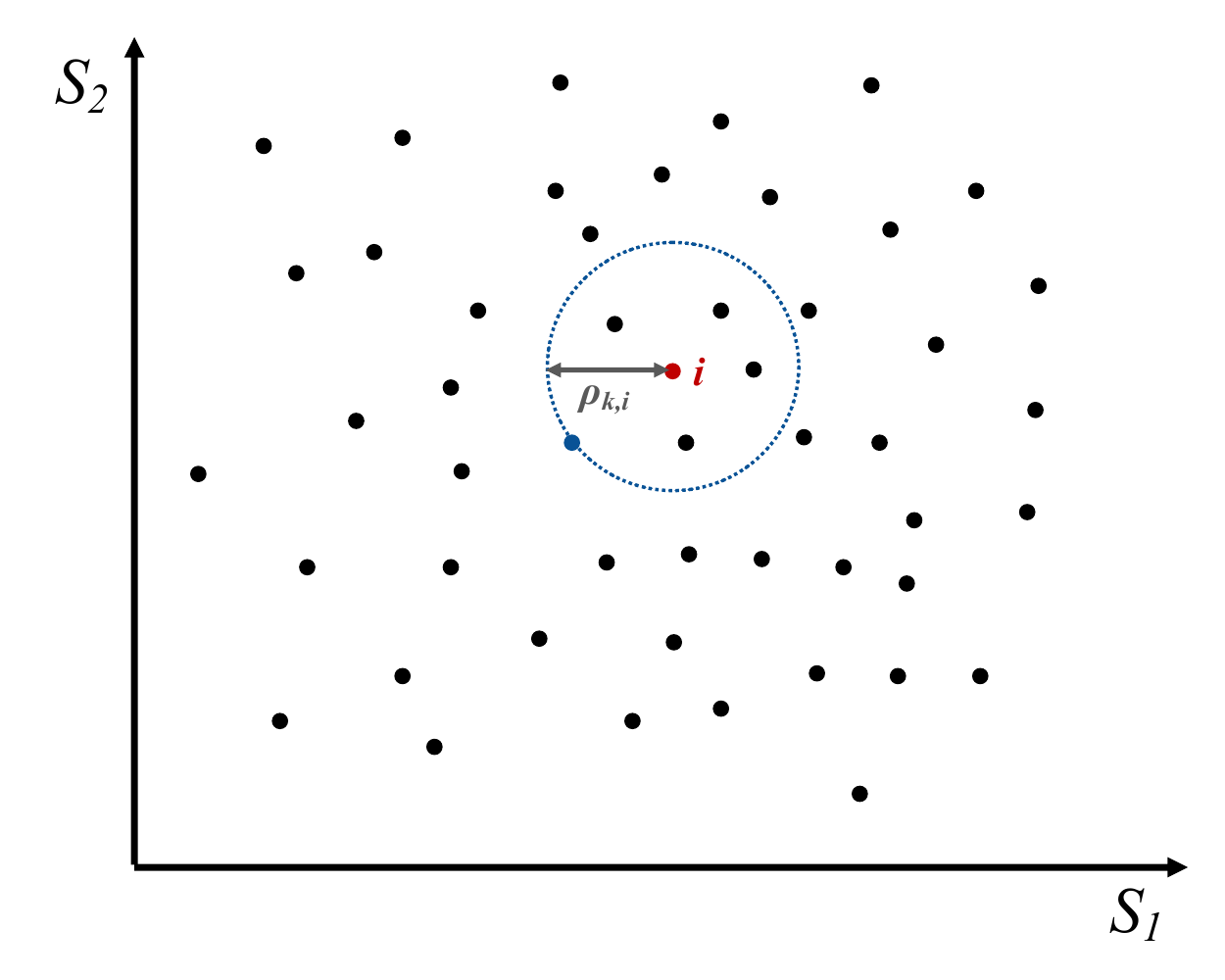}
	\caption{}
	\label{fig:step_1}
\end{subfigure}
\begin{subfigure}[t]{.5\textwidth}
\centering
	\includegraphics[scale=.5]{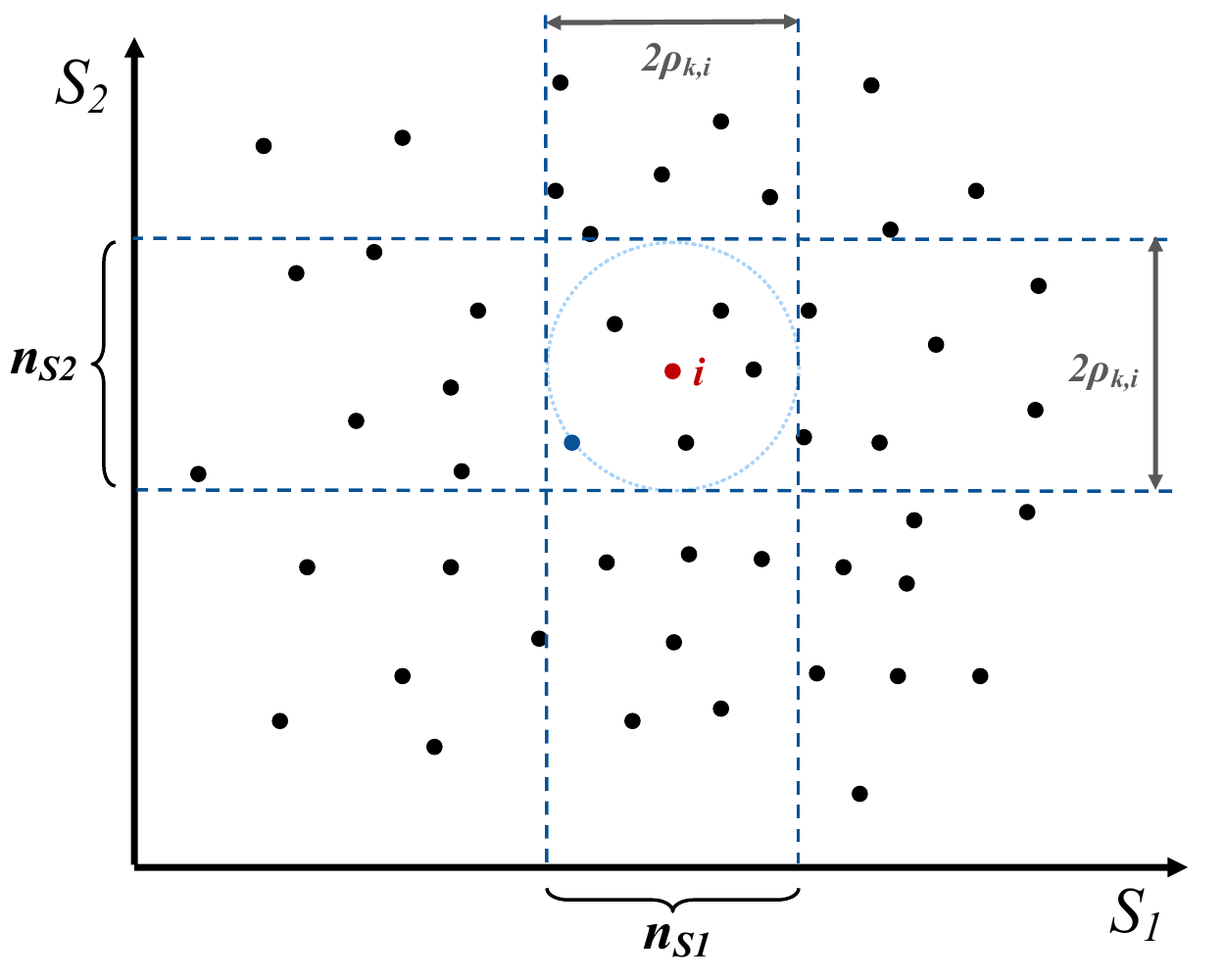}
	\caption{}
	\label{fig:step_2}
\end{subfigure}
\caption{The method for estimating information theoretic measures. \textbf{(a)} Step 1: \textbf{Query} for the distance to the $k$th nearest neighbor of each point $i$ in the space $\mathcal{X}$. \textbf{(b)} Step 2: \textbf{Inquire} for the number of points lying within the $\rho_{k,i}$-neighborhood of each point $i$ in the subspaces of $\mathcal{X}$ including itself.}
\label{fig:estimator_schematic}
\end{figure}

\section{Multivariate Mutual Information}
In \cite{chan2015multivariate}, Multivariable Mutual Information (MMI) is defined as follows: let $\Pi(\mathcal{X})$ be the collection of all possible partitions of $\mathcal{P}$ which split $\mathcal{X}$ into at least two non-empty disjoint subsets. For any partition $\mathcal{P}\in\Pi(\mathcal{X})$, the product distribution $\Pi_{C\in\mathcal{P}}\mathbb{P}_{X_C}$ specifies an independence relation, i.e., $X_{C}$'s are treated as agglomerated random variables and are mutually independent. Given a particular partition, define an information measure $I_{\mathcal{P}}(X)$ as :
\begin{align}
I_{\mathcal{P}}(X) = \frac{1}{\mathcal{P}-1} D (\mathbb{P}_{X}\parallel \Pi_{C\in\mathcal{P}}\mathbb{P}_{X_C})
\end{align}

Then, MMI is defined as:
\begin{align}
\text{MMI(X)}  = \min_{\mathcal{P}\in\Pi(\mathcal{X})} I_{\mathcal{P}}(X)
\end{align}

This can cast as a functional of our Graph Divergence Measure by choosing for every partition, $\mathcal{P}\in\Pi(\mathcal{X})$, a DAG $\mathcal{G}_{\mathcal{P}}$ with all $X_{C}$'s forming an aggregate node but disconnected from each other and thus inducing a measure $\gls{graph_dist}_X^{\mathcal{P}}$. Thus,
\begin{align}
I_{\mathcal{P}}(X) = \frac{1}{\mathcal{P}-1} \mathbb{GDM}(\mathbb{P}_X\parallel  \gls{graph_dist}_X^{\mathcal{P}})
\end{align}
which implies, 
\begin{align}
\text{MMI(X)}  =  \min_{\mathcal{P}\in\Pi(\mathcal{X})} \mathbb{GDM}(\mathbb{P}_X\parallel  \gls{graph_dist}_X^{\mathcal{P}})
\end{align}

\section{Directed Information}

In this section, we will derive the expression of the directed information from $X^T=\left( X_1, \ldots, X_T \right)$ to $Y^T=\left( Y_1, \ldots, Y_T \right)$ in terms of two graph divergence measures. For the simplicity of notations and logic we will only do it assuming $X$ and $Y$ are discrete. However, it's easily extendable to the case of mixture distributions using the notion of Radon-Nikodym derivative. From the definition of the directed information we have:

\begin{eqnarray}
& & I \left( X^T \rightarrow Y^T \right) \\
& = & \sum_{t=1}^T I \left( X^t  ; Y_t \middle| Y_{t-1} \right) \\
& = & \sum_{t=1}^T \sum_{x^t,y^t} \left( \mathbb{P}_{X^tY^t}(x^t,y^t) \log \frac{ \mathbb{P}_{Y_t | X^t,Y^{t-1}}(y_t | x^t,y^{t-1} ) } {  \mathbb{P}_{Y_t |Y^{t-1}} (y_t | y^{t-1}) } \right) \\
& = & \sum_{x^T,y^T} \mathbb{P}_{X^TY^T}(x^T,y^T) \sum_{t=1}^T \log \frac{ \mathbb{P}_{Y_t | X^t,Y^{t-1}}(y_t | x^t,y^{t-1} ) } { \mathbb{P}_{Y_t |Y^{t-1}} (y_t | y^{t-1}) } \\
& = &  \sum_{x^T,y^T} \mathbb{P}_{X^TY^T}(x^T,y^T) \log \frac{ \prod_{t=1}^T\mathbb{P}_{Y_t | X^t,Y^{t-1}}(y_t | x^t,y^{t-1} ) }
{ \prod_{t=1}^T\mathbb{P}_{Y_t |Y^{t-1}} (y_t | y^{t-1}) } \\
& = &  \sum_{x^T,y^T} \mathbb{P}_{X^TY^T}(x^T,y^T) \log \frac{ \prod_{t=1}^T\mathbb{P}_{Y_t | X^t,Y^{t-1}}(y_t | x^t,y^{t-1} ) }
{ \mathbb{P}_{Y^T} (y^T) } \\
& = &  \sum_{x^T,y^T} \mathbb{P}_{X^TY^T}(x^T,y^T) \log \frac{ \prod_{t=1}^T \mathbb{P}_{Y_t | X^t,Y^{t-1}}(y_t | x^t,y^{t-1} )\mathbb{P}_{X_t | X^{t-1}}(x_t | x^{t-1} ) }
{ \mathbb{P}_{X^T} (x^T) \mathbb{P}_{Y^T} (y^T) } \\
& = & \sum_{x^T,y^T} \mathbb{P}_{X^TY^T}(x^T,y^T) \log \frac{ \mathbb{P}_{X^TY^T}(x^T,y^T) } { \mathbb{P}_{X^T} (x^T) \mathbb{P}_{Y^T} (y^T) } \nonumber \\
& & - \sum_{x^T,y^T} \mathbb{P}_{X^TY^T}(x^T,y^T) \log \frac{ \mathbb{P}_{X^TY^T}(x^T,y^T) }{ \prod_{t=1}^T \mathbb{P}_{Y_t | X^t,Y^{t-1}}(y_t | x^t,y^{t-1} )\mathbb{P}_{X_t | X^{t-1}}(x_t | x^{t-1} ) }\\
& = & D \left( \mathbb{P}_{X^TY^T} \| \mathbb{P}_{X^TY^T}^I \right)  - D \left( \mathbb{P}_{X^TY^T} \| \mathbb{P}_{X^TY^T}^C \right) \\
& = & \gls{graph_distance_measure}\Big( (X^T,Y^T), \mathcal{G}_I \Big) - \gls{graph_distance_measure}\Big( (X^T,Y^T), \mathcal{G}_C \Big)
\end{eqnarray}

The distributions $\mathbb{P}_{X^TY^T}^I= \mathbb{P}_{X^T} \mathbb{P}_{Y^T}$ and $\mathbb{P}_{X^TY^T}^C = \prod_{t=1}^T \mathbb{P}_{Y_t | X^t,Y^{t-1}}\mathbb{P}_{X_t | X^{t-1}}$ represent the \textit{independent} distribution between $X^T$ and $Y^T$, and the \textit{causal} distribution from $X^T$ to $Y^T$ respectively.

\section{Variational Representation of Graph Divergence Measure}

%
The first part follow from the following property of KL divergence: 
 \begin{align}
D(\mathbb{P}_X \| \mathbb{Q}_X) = D(\mathbb{P}_X \| \gls{graph_dist}_X)+ \sum_{j}^d D(\mathbb{P}_{X_j|X_{Pa(j)}} \| \mathbb{Q}_{X_j|X_{Pa(j)}} | \mathbb{P}_{X_1^{j-1}}) 
\end{align}

The second part follows from the standard Donsker-Varadhan characterizations of vanilla divergence as in \cite{itlectures}.

\section{Investigating assumptions in Theorem \ref{thm:mean_convergence}}

In this section, we will investigate the validity of the assumptions for various types of distributions. In particular, we will investigate the different types of distributions we had mentioned: (1) finitely discrete; (2) continuous; (3) finitely discrete over some dimensions while continuous over others; (4) a mixture of the previous cases; (5) has a joint density supported over a lower dimensional manifold. As we had mentioned, these cases represent almost all the real world data. We note that the estimator is well defined only when $\mathcal{X} = \mathbb{R}^{d_X}$; i.e. all the alphabets be in real space.

We will examine the Assumptions \ref{assumptions}.\ref{assumption_finitediscrete} and \ref{assumptions}.\ref{assumption_logintegrable} for each case separately. The Assumption \ref{assumptions}.\ref{assumption_kn} is a parametric assumption related to the convergence of the algorithm and is not directly related to the distribution of the data. Jointly considered with the Assumption  \ref{assumptions}.\ref{assumption_finitediscrete}, it controls the boundary effects between the continuous and the discrete regions.

\subsection{Finitely discrete distribution}\label{subsec:assumptions_discrete}
For a \textit{finitely discrete} distribution the Assumption \ref{assumptions}.\ref{assumption_finitediscrete} holds by definition. The Assumption \ref{assumptions}.\ref{assumption_logintegrable} trivially holds since the size of sample space $| \mathcal{X} |$ is finite. 

\subsection{Continuous distribution}\label{subsec:assumptions_continuous}

For a continuous distribution, the Assumption \ref{assumptions}.\ref{assumption_finitediscrete} holds since there is no discrete component. 

The distribution is absolutely continuous with respect to the \textit{Lebesgue} measure $\lambda$ meaning that $\mathbb{P}_X(A)=0$ for any subset $A \subset \mathcal{X}$ implies $\lambda(A)=0$. Thus we can conclude that there exists a density $g: \mathcal{X} \rightarrow \mathbb{R}^+$ such that $\mathbb{P}_X(A) = \int_A g\hspace{.1cm} d\lambda$. Naturally $\mathbb{P}_X(\mathcal{X})=\int_\mathcal{X} g\hspace{.1cm} d\lambda = 1$. 

Furthermore if the density function $g$ is bounded everywhere and the variable has a bounded support, the Assumption \ref{assumptions}.\ref{assumption_logintegrable} is fulfilled.

\subsection{Finitely discrete over some dimensions while continuous over others}

In this case, the variable set $X$ can be decomposed to two sets $X^D$ and $X^C$ representing the finitely discrete and continuous dimensions respectively. So for any realization of the discrete dimensions $X^D=x^D$ the probability mass function $h_{X^D}(x^D)$ exists, and a conditional density $g_{X^C|X^D}(x^C|x^D)$ is well defined. We can define an auxiliary function $P_{X^C|X^D}(x^C,r|x^D) \equiv \mathbb{P}_{X^C|X^D} \left\{ a \in \mathcal{X}^C : \| a - x^C \|_{\infty} \leq r \middle| X^D=x^D \right\}$. Thus we can write \begin{equation}
P_X(x^C,x^D,r) = \mathbb{P}_X \left\{ (a,b) \in \mathcal{X} : b = x^D, \| a - x^C \|_{\infty} \leq r \right\} = P_{X^C|X^D}(x^C,r|x^D) h_{X^D}(x^D)
\end{equation}

We will show that if for any $x^D$ the continuous distribution over $X^C$ satisfies the Assumptions \ref{assumptions} as discussed in the Section \ref{subsec:assumptions_continuous}, then we can see that $\mathbb{P}_X$ will also satisfy the assumptions.

Let's define $f_{x^D}(x^C)$ for any fixed $x^D$ as:
\begin{eqnarray}
f_{x^D}(x^C)= \lim_{r \rightarrow 0} P_{X^C|X^D}(x^C,r|x^D) 
\prod_{l=1}^d \frac{P_{{\gls{pa_xl_var}^C|\gls{pa_xl_var}^D}}\left( \gls{pa_xl_val}^C ,r \middle| \gls{pa_xl_val}^D \right)}
{P_{{\gls{pa_n_xl_var}^C|\gls{pa_n_xl_var}^D}}\left( \gls{pa_n_xl_val}^C, r \middle| \gls{pa_n_xl_val}^D \right)}
\end{eqnarray}
Assuming the limit exists everywhere. Thus the Radon-Nikodym Derivative $f$ can be written as:
\begin{eqnarray}
f(x) & = & \lim_{r \rightarrow 0} P_X(x^C,x^D,r) \prod_{l=1}^d \frac{P_{\gls{pa_xl_var}}\left( \gls{pa_xl_val}^C, \gls{pa_xl_val}^D ,r \right)}{P_{\gls{pa_n_xl_var}}\left(  \gls{pa_n_xl_val}^C, \gls{pa_n_xl_val}^D, r \right)} \\
& = &
\lim_{r \rightarrow 0} P_{X^C|X^D}(x^C,r|x^D) 
\prod_{l=1}^d \frac{P_{{\gls{pa_xl_var}^C|\gls{pa_xl_var}^D}}\left( \gls{pa_xl_val}^C ,r \middle| \gls{pa_xl_val}^D \right)}
{P_{{\gls{pa_n_xl_var}^C|\gls{pa_n_xl_var}^D}}\left( \gls{pa_n_xl_val}^C, r \middle| \gls{pa_n_xl_val}^D \right)} \nonumber \\
& & \times
h_{X^D}(x^D) \prod_{l=1}^d \frac{h_{\gls{pa_xl_var}}\left(\gls{pa_xl_val}^D \right)}{h_{\gls{pa_n_xl_var}}\left( \gls{pa_n_xl_val}^D \right)} \\
& = &
f_{x^D}(x^C) \times h_{X^D}(x^D) \prod_{l=1}^d \frac{h_{\gls{pa_xl_var}}\left( \gls{pa_xl_val}^D \right)}{h_{\gls{pa_n_xl_var}}\left(  \gls{pa_n_xl_val}^D \right)}
\end{eqnarray}

Therefore for Assumption \ref{assumptions}.\ref{assumption_logintegrable} we have:
\begin{eqnarray} 
\int_{\mathcal{X}} \left| \log f(x) \right| d\mathbb{P}_X &\leq& \sum_{x^D} h_{X^D}(x^D) \int_{\mathcal{X}^C} \left| \log f_{x^D}(x^C) \right| g_{X^C|X^D}(x^C|x^D) dx^C \nonumber \\
& & + \sum_{x^D} h_{X^D}(x^D) \left| \log h_{X^D}(x^D) \prod_{l=1}^d \frac{h_{\gls{pa_xl_var}}\left( \gls{pa_xl_val}^D \right)}{h_{\gls{pa_n_xl_var}}\left(  \gls{pa_n_xl_val}^D \right)} \right|
\end{eqnarray}
In which we used the fact that $\mathbb{E} \left[ \log f(x) \right] = \mathbb{E}_{X^D} \left[ \mathbb{E} \left[ \log f(x^C,x^D) \middle| X^D=x^D \right] \right]$. The first term above is upper-bounded since we assumed that continuous distribution over $X^C$ satisfies the Assumption \ref{assumptions}.\ref{assumption_logintegrable} for any $x^D$. The second term is upper-bounded since it's a finitely discrete distribution as discussed in Section \ref{subsec:assumptions_discrete}. Thus $\int_{\mathcal{X}} \left| \log f(x) \right| d\mathbb{P}_X < \infty$ and the Assumption \ref{assumptions}.\ref{assumption_logintegrable} holds.

\subsection{A mixture of the previous cases}

In this case, we assume that the probability can be described as a linear combination of a continuous and a discrete distribution, i.e. without loss of generality we can assume that for any subset $A \subset \mathcal{X}$ the distribution can be described as $\mathbb{P}(A)= \alpha_C \mathbb{P}^C_X(A) + \alpha_D \mathbb{P}^{D}_X(A)$ in which $\alpha_C + \alpha_D = 1$. We note that the $\mathbb{P}_X^D$ does not represent a mass function here since the mass function is only defined over a discrete alphabet $\mathcal{X}^D= \{ x_1, ..., x_m \} \subset \mathcal{X}$ while $\mathbb{P}^D_X$ needs to admit the continuous domain $\mathcal{X}$. Thus we relate $\mathbb{P}^D_X$ to the mass function $h^D_X(.)$ as $\mathbb{P}^D(A) = \sum_{x \in \mathcal{X}^D} \boldsymbol{1}_{\{ x \in A\}} h_X^D(x)$. Furthermore we can see that for any $x \in \mathcal{X}^D$ and for $r$ small enough: 
\begin{equation} P_X^D(x,r) \equiv \mathbb{P}^D_X \Big( B_r(x) \Big) = h_X^D(x) \end{equation}
If $x \notin \mathcal{X}^D$ then for small enouth $r$ we have $P_X^D(x,r)=0$.

The Assumption \ref{assumptions}.\ref{assumption_finitediscrete} holds since the discrete distribution $\mathbb{P}^D_X$ is finite, and hence the number of total discrete points will be finite.

For the Assumption \ref{assumptions}.\ref{assumption_logintegrable} we have:
\begin{equation} \int_{\mathcal{X}} \left| \log f(x) \right| d\mathbb{P}_X = \alpha_C \int_{\mathcal{X}} \left| \log f(x) \right| d\mathbb{P}^{C}_X 
+ \alpha_D \sum_{x \in \mathcal{X}^D} h^{D}_X(x) \left| \log f(x) \right| \end{equation}
and if the continuous distribution complies with the assumptions mentioned in Section~\ref{subsec:assumptions_continuous} then the term above is upper-bounded and the Assumption~\ref{assumptions}.\ref{assumption_logintegrable} holds.

\subsection{A distribution with a joint density supported over a lower dimensional manifold}

This case simply means that the probability distribution $\mathbb{P}_X$ in $\mathcal{X}$ can be mapped to a probability distribution $\mathbb{P}_Y$ in a lower-dimensional space $\mathcal{Y}$ where $d_Y < d_X$, via a one-to-one continuous function $h: \mathcal{X} \rightarrow \mathcal{Y}$.
If the lower-dimnesional distribution $\mathbb{P}_Y$ is continuous complying with the properties discussed in the Section \ref{subsec:assumptions_continuous}, then it will preserve all the properties through the inverse mapping $h^{-1}: \mathcal{Y} \rightarrow \mathcal{X}$ and $\mathbb{P}_{X}$ hence $\mathbb{P}_X$ will satisfy the Assumptions \ref{assumptions}.
If the distribution has finite discrete components either as a discrete or as a mixture distribution, the finiteness of the components will be preserved as well and hence $\mathbb{P}_X$ will satisfy the Assumptions \ref{assumptions}.

Therefore this category of distributions will satisfy the Assumptions \ref{assumptions} if the distribution $\mathbb{P}_Y$ satisfies them.

\section{Consistency Proofs}

\subsection{Proof of Theorem \ref{thm:mean_convergence}} \label{sec:mean_convergence_proof}

First let’s generalize the definition of $P_X(x,r)$ in Equation \ref{eq:P_xr} to any subset $S \subseteq X$, i.e. for any point $s \in S$ we define:
\begin{equation} P_S(s,r) = \mathbb{P}_S \big\{ a \in \mathcal{S} : \| a- s \|_{\infty} \leq r \big\} = \mathbb{P}_S \Big\{ B_r(s) \Big\}\end{equation}
Thus $P_S(s,r)$ is the probability of an $\ell_\infty$ ball of radius $r$ centered at $s$, or equivalently, a hypercube with the edge length of $2r$ centered at the point $s$.

To prove the asymptotic unbiasedness of the estimator, we will first write the Radon-Nikodym derivative in an explicit form via the following lemma. 

\begin{lemma}\label{lemma:RN_convergence}
For almost every $x \in \mathcal{X}$:
\begin{equation} \frac{d \mathbb{P}_X}{d \gls{graph_dist}_X}(x) = f(x) = \lim_{r\rightarrow 0} P_X(x,r) \prod_{l=1}^d \frac{P_{\gls{pa_xl_var}}\left( \gls{pa_xl_val} ,r \right)}{P_{\gls{pa_n_xl_var}}\left(  \gls{pa_n_xl_val} ,r \right)} \end{equation} 
\end{lemma}
\begin{proof}
Please see the section \ref{sec:lemma_RN_convergence_proof}.
\end{proof}

Now notice that $\widehat{\gls{graph_distance_measure}}^{(N)}(X,\mathcal{G})=\frac{1}{N} \sum_{i=1}^N \zeta_i$ in which all the $\zeta_i$'s are identically distributed. Thus we have $\mathbb{E} [ \widehat{\gls{graph_distance_measure}}^{(N)}(X,\mathcal{G}) ] = \mathbb{E} [ \zeta_1 ] $. Therefore, the bias can be written as:
\begin{eqnarray}
\left| \mathbb{E} \left[ \widehat{\gls{graph_distance_measure}}^{(N)}(X,\mathcal{G}) \right] - \gls{graph_distance_measure}(X,\mathcal{G}) \right| & = &  \left| \mathbb{E}_X\left[ \mathbb{E}[ \zeta_1 | X ] \right] - \int_{\mathcal{X}} \log f(X) d\mathbb{P}_X \right| \\
 & \leq & \int_{\mathcal{X}} \left| \mathbb{E}[ \zeta_1 | X ] -  \log f(X)  \right| d\mathbb{P}_X
\end{eqnarray}

Now we will give upper bounds for $\left| \mathbb{E}[ \zeta_1 |X] -  \log f(X)  \right|$ for every $x \in \mathcal{X} $. Similar to the technique used by \cite{gao2017mixture}, we divide the domain of $\mathcal{X}$ into three parts as $\mathcal{X} = \Omega_1 \bigcup \Omega_2 \bigcup \Omega_3 $ where:
\begin{itemize}
\item $\Omega_1 = \{x: f(x)=0 \}$
\item $\Omega_2 = \{x: f(x)>0, P_X(x,0)>0 \}$
\item $\Omega_3 = \{x: f(x)>0, P_X(x,0)=0 \}$
\end{itemize}

For each of the domains, we show that $\lim_{N \rightarrow \infty} \int_{\Omega_i} \left| \mathbb{E}[ \zeta_1 | X=x ] -  \log f(x)  \right| d\mathbb{P}_X = 0$.

\vspace{.1in}
\textbf{For $\boldsymbol{x \in \Omega_1}$:}

The probability of $\Omega_1$ is zero with respect to $\mathbb{P}_X$, since:
\begin{equation} \mathbb{P}_X(\Omega_1)=\int_{\Omega_1} d\mathbb{P}_X = \int_{\Omega_1} f(x)d\gls{graph_dist}_X = \int_{\Omega_1} 0 d\gls{graph_dist}_X = 0\end{equation} 
In which the second equality is due to Lemma \ref{lemma:RN_convergence}. Thus $\int_{\Omega_1} \left| \mathbb{E}[ \zeta_1 | X=x ] -  \log f(x)  \right| d\mathbb{P}_X = 0$.

\vspace{.1in}
\textbf{For $\boldsymbol{x \in \Omega_2}$:}

In this case $f(x)$ is obviously the same as $P_X(x,0) \prod_{l=1}^d \frac{P_{\gls{pa_xl_var}}\left( \gls{pa_xl_val} ,0 \right)}{P_{\gls{pa_n_xl_var}}\left( \gls{pa_n_xl_val} ,0 \right)}$. We will first show that the probability of the $k$-nearest neighbor distance $\rho_{k,1}$ being non-zero is small, which means we will use the number of samples being equal to $x$ as $\tilde{k}_i$, and we will show that the mean of the estimate $\zeta_1$ is close to $\log f(x)$.

We notice that for $x$, the probability of $\rho_{k,1}>0$ is equal to the probability that $x$ is observed at most $k-1$ times. So it can be upper bounded as:

\begin{eqnarray}
& & \mathbb{P} \Big( \rho_{k,1}>0 \Big| X = x \Big) \\
&=& \sum_{m=0}^{k-1} \left( \begin{array}{c} N-1 \\ m \end{array} \right) P_X(x,0)^m \Big(1-P_X(x,0) \Big)^{N-1-m} \\
& \leq & \sum_{m=0}^{k-1} N^m \Big(1-P_X(x,0) \Big)^{N-k} \\
& \leq & k N^k \Big(1-P_X(x,0) \Big)^{N-k} \\
& \leq & k N^k e^ {-(N-k)P_X(x,0)} 
\end{eqnarray} 

Now let's consider the case when $\rho_{k,1}=0$. We can write the term $\zeta_1$ in the form: 
\begin{equation}\zeta_1 = \psi(\tilde{k}_1) + \sum_{l=1}^d \left(  \mathbf{1}_{\{ \gls{pa_xl_var} \neq \emptyset \} } \log (n_{\gls{pa_xl_var}}^{(1)}+1) - \log (n_{\gls{pa_n_xl_var}}^{(1)}+1) \right) + K_N \nonumber \end{equation}

The term $K_N$ depends only on $N$ and the structure of the Bayesian model $\gls{graph_dist}_X$ and is independent of the observed samples, and in general is equal to:
\begin{equation}K_N =  -\log C_{d_X} + \sum_{l=1}^d \left( \log C_{d_{\gls{pa_n_xl_var}}}-\log C_{d_{\gls{pa_xl_var}}} \right) + \left( \sum_{l=1}^d \mathbf{1}_{\{ \gls{pa_xl_var} = \emptyset \}}-1 \right) \log N. \end{equation}

In which the terms $C_{d_S}$ indicate the volume of a unit ball in an $d_S$-dimensional space $\mathcal{S}$. Since we are using $\ell_\infty$-norm in our algorithm and proofs, all $C_{d_S}$ terms will be equal to $1$ and hence $K_N = \left( \sum_{l=1}^d \mathbf{1}_{\{ \gls{pa_xl_var} =\emptyset \}}-1 \right) \log N$ as it appeared in Algorithm \ref{algm:gdm}.

%
%
%
%
Then for the case of $\rho_{k,1}=0$ we can write:

\begin{eqnarray}
& & \Big| \mathbb{E}[ \zeta_1 | X=x, \rho_{k,1} =0 ] - \log f(x) \Big| \\
& = & \bigg| \mathbb{E} \bigg[ \psi(\tilde{k}_i) + \sum_{l=1}^d \left( \mathbf{1}_{ \{ \gls{pa_xl_var} \neq \emptyset \} } \log (n_{\gls{pa_xl_var}}^{(1)}+1) - \log (n_{\gls{pa_n_xl_var}}^{(1)}+1) \right) \nonumber\\ 
& & + \left( \sum_{l=1}^d \mathbf{1}_{ \{ \gls{pa_xl_var}= \emptyset \} } -1 \right) \log N \bigg| X=x, \rho_{k,1} =0 \bigg] \nonumber \\
& & - \log P_X(x,0) \prod_{l=1}^d \frac{P_{\gls{pa_xl_var}}\left( \gls{pa_xl_val} ,0 \right)}{P_{\gls{pa_n_xl_var}}\left( \gls{pa_n_xl_val} ,0 \right)} \bigg| \\
& \leq & \left| \mathbb{E}[ \psi(\tilde{k}_1) | X=x, \rho_{k,1} =0 ] - \log N P_X(x,0) \right| \label{eq:omega2_term1} \\
& & + \sum_{l=1}^d \left| \mathbb{E}[ \log (n_{\gls{pa_n_xl_var}}^{(1)}+1) |X=x, \rho_{k,1} =0 ] - \log N P_{\gls{pa_n_xl_var}}\left( \gls{pa_n_xl_val} ,0 \right) \right| \label{eq:omega2_term2} \\
& & + \sum_{l=1}^d \mathbf{1}_{ \{ \gls{pa_xl_var} \neq \emptyset \} } \left| \mathbb{E}[ \log (n_{\gls{pa_xl_var}}^{(1)}+1) | X=x, \rho_{k,1} =0 ] - \log N P_{\gls{pa_xl_var}}\left(  \gls{pa_xl_val} ,0 \right) \right| \hspace{.4in} \label{eq:omega2_term3}
\end{eqnarray}

We notice that $\tilde{k}_1$ is the number of samples among $\{ x^{(i)} \}_{i=1}^N$ such that $X_i = x$, where each sample is independently equal to $x$ with probability $P_X(x,0)$. Therefore the distribution of $\tilde{k}_1$ is $\text{Bino}\left( N,P_{X}(x,0) \right)$. Similarly, for any $S \subset X$, $n_{S,1}+1$ is the number of samples among $\{ x^{(i)} \}_{i=1}^N$ such that $s^{(i)}=s$; i.e. projection of $x^{(i)}$ over $\mathcal{S}$ is equal to the projection of $x$ over $\mathcal{S}$. Thus $n_{S,1}+1 \sim \text{Bino}\left( N,P_S(s,0) \right)$. In addition to that, the event $\rho_{k,1}$=0 is equivalent to $\tilde{k}_1 \geq k$. Thus to upperbound term \ref{eq:omega2_term1} and any of the individual terms inside the summations of terms \ref{eq:omega2_term2} and \ref{eq:omega2_term3} we propose the lemma below:

\begin{lemma}\label{lemma:omega2_upperbound}
If $X$ is distributed as $\text{Bino}(N,p)$ and $m \geq 0$ , then: 
\begin{eqnarray} \left| \mathbb{E} \left[ \log(X+m) \middle| X \geq k \right] - \log (Np) \right| \leq U(k,N,m,p) \end{eqnarray} 
Where $U(k,N,m,p)$ is given by:
\begin{equation}
U(k,N,m,p) = \max \left\{ \left| \log \left( \frac{1+\frac{m}{Np}}{1 - \exp\left( -2 \frac{(Np-k)^2}{N} \right)} \right) \right|
 , \frac{1}{1 - \exp\left( -2 \frac{(Np-k)^2}{N} \right)} \frac{3}{2Np} \right\} \end{equation}
\end{lemma}

\begin{proof}
Please see Section \ref{sec:lemma_omega2_upperbound_proof}.
\end{proof}

\begin{remark}
Since the Assumption \ref{assumptions}.\ref{assumption_kn} states that $k/N \rightarrow 0$ as $N \rightarrow \infty$, then $(Np-k)^2/N= N (p-k/N)^2 \rightarrow \infty$, and the upperbound $U(k,N,m,p)$  will converge to $0$ as $N \rightarrow \infty$ for any $p$.
\end{remark}



From Lemma \ref{lemma:omega2_upperbound} we have:


\begin{eqnarray} & & \Big| \mathbb{E}[ \log (n_S^{(1)}+1) | X=x, \rho_{k,1} =0 ] - \log N P_S(s,0) \Big| \\
& = & \Big| \mathbb{E}[ \log (n_S^{(1)}+1) | X=x, n_S^{(1)}+1 \geq k ] - \log N P_S(s,0) \Big| \leq U \Big( k,N,0,P_S(s,0) \Big)  \end{eqnarray}

For any of the individual terms inside the summations of terms \ref{eq:omega2_term2} and \ref{eq:omega2_term3}. For the term \ref{eq:omega2_term1} we notice that $| \psi(\tilde{k}_1)-\log(\tilde{k}_1)| \leq 1/\tilde{k}_1 \leq 1/k$ \cite{bernardo1976algorithm}. Thus this term can be bounded similarly by $ U \left( k,N,0,P_X(x,0) \right)+1/k$. By combining all these bounds, we obtain:
\begin{eqnarray}
& & \Big| \mathbb{E}[ \zeta_1 | X=x, \rho_{k,1} =0 ] - \log f(x) \Big| \\
& \leq &  \sum_{l=1}^d \left( U \Big( k,N,0,P_{\gls{pa_xl_var}}(\gls{pa_xl_val},0) \Big) +  U \Big( k,N,0, P_{\gls{pa_n_xl_var}}(\gls{pa_n_xl_val},0) \Big) \right) \nonumber \\
& & + U \Big( k,N,0,P_X(x,0) \Big) + \frac{1}{k}
\end{eqnarray}

Assumption \ref{assumptions}.\ref{assumption_finitediscrete} implies that discrete probabilities and hence $U \Big( k,N,0,p \Big)$ terms are bounded; i.e. there exists a $\hat{p}$ such that for any $x \in \Omega_2$:
\begin{equation}
\left\{ \begin{array}{lr}
U \Big( k,N,0,P_X(x,0) \Big) \leq U \Big( k, N, 0, \hat{p} \Big) & \\
U \Big( k, N, 0, P_{\gls{pa_xl_var}}(\gls{pa_xl_val},0) \Big) \leq U \Big( k, N, 0, \hat{p}\Big) & \text{ for } l=1, \ldots, d \\
U \Big( k, N, 0, P_{\gls{pa_n_xl_var}}(\gls{pa_n_xl_val},0),0) \Big) \leq U \Big( k, N, 0, \hat{p} \Big) & \text{ for } l=1, \ldots, d \\
\end{array} \right.
\end{equation}

Therefore:
\begin{eqnarray}
\Big| \mathbb{E}[ \zeta_1 | X=x, \rho_{k,1} =0 ] - \log f(x) \Big| \leq  \left( 2d +1 \right) U \Big( k, N, 0, \hat{p} \Big) + \frac{1}{k}
\end{eqnarray}

If we combine it with the case of $\rho_{k,i}>0$, we obtain that:
\begin{eqnarray}
& & \Big| \mathbb{E}[ \zeta_1 |X=x] - \log f(x) \Big|\\
& \leq & \Big| \mathbb{E}[ \zeta_1 | X=x, \rho_{k,1} >0 ] - \log f(x) \Big| \times \mathbb{P} (\rho_{k,1} >0 ) \nonumber\\
& & + \Big| \mathbb{E}[ \zeta_1 | X=x, \rho_{k,1} =0 ] - \log f(x) \Big| \times \mathbb{P} (\rho_{k,1} =0 ) \\
& \leq & \Big( (2d + 1 ) \log N + \left| \log f(x) \right| \Big) k N^k e^ {-(N-k)P_X(x,0)} + \left( 2d +1 \right) U \Big( k, N, 0, \hat{p} \Big) + \frac{1}{k} \hspace{.5in}
\end{eqnarray}

Where we used the fact that $| \zeta_1 | \leq \left( 2d+ 1 \right) \log N$. Integrating over $\Omega_2$, we can write:
\begin{eqnarray}
& & \int_{\Omega_2} \Big| \mathbb{E}[ \zeta_1 |X=x] - \log f(x) \Big| d\mathbb{P}_X\\
& \leq & \left((2d + 1 ) \log N + \int_{\Omega_2} \left| \log f(x) \right| d\mathbb{P}_X \right) k N^k e^ {-(N-k) \inf_{x\in \Omega_2} P_X(x,0)} \\
& & +  \left( 2d +1 \right) U \Big( k, N, 0, \hat{p} \Big) + \frac{1}{k}
\end{eqnarray}

By Assumption \ref{assumptions}.\ref{assumption_kn}, $k$ goes to infinity as $N$ goes to infinity, so $1/k$ vanishes as $N$ grows boundlessly. The term $U \Big( k, N, 0, \hat{p} \Big)$ will also converge to zero as we saw. From assumption \ref{assumptions}.\ref{assumption_logintegrable}, $\int_{\Omega_2} \left| \log f(x) \right| d\mathbb{P}_X < +\infty$. Thus the first term also converges to $0$ as $N \rightarrow \infty$. Thus:

\begin{equation} \lim_{N \rightarrow \infty}  \int_{\Omega_2} \Big| \mathbb{E}[ \zeta_1 | X=x] - \log f(x) \Big| d\mathbb{P}_X = 0 \end{equation} 

\vspace{.1in}
\textbf{For $\boldsymbol{x \in \Omega_3}$:}

In this case, $P_X(x,r)$ is a monotonic function of $r$ such that $P_X(x,0)=0$ and $\lim_{r \rightarrow \infty} P_X(x,r)=1$. Hence we can view $\log P_X(x,r) \prod_{l=1}^d \frac{P_{\gls{pa_xl_var}}\left( \gls{pa_xl_val} ,r \right)}{P_{\gls{pa_n_xl_var}}\left( \gls{pa_n_xl_val} ,r \right)}$ as a function of $P_X(x,r)$ and it converges to $\log f(x)$ as $P_X(x,r) \rightarrow 0$, for almost every $x$. Since $\mathbb{P}_X(\Omega_3) \leq 1 < \infty$ and we know $\int_{\Omega_3} \left| \log f(x) \right| d\mathbb{P}_X < \infty$, By Egorov's theorem, for any $ \epsilon>0$, there exists a subset $ E \subseteq \Omega_3$ with $\mathbb{P}_X(E) < \epsilon$ and $\int_E \left| \log f(x) \right| d\mathbb{P}_X < \epsilon$, such that the term $\log P_X(x,r) \prod_{l=1}^d \frac{P_{\gls{pa_xl_var}}\left( \gls{pa_xl_val} ,r \right)}{P_{\gls{pa_n_xl_var}}\left( \gls{pa_n_xl_val} ,r \right)}$ converges uniformly on $\Omega_3 \setminus E$ as $P_X(x,r) \rightarrow 0$. Now let's assume $\epsilon_N$ is a sequence converging to $0$. Consequently, The corresponding sets $E_N$ will also create a sequence. For any fixed $N$ and for $x \in E_N$, we notice that $| \zeta_1 | \leq \left( 2 d+ 1 \right) \log N$, so we have:
\begin{eqnarray}
& & \int_{E_N} \Big| \mathbb{E}[\zeta_1|X=x] - \log f(x) \Big| d\mathbb{P}_X \\
& \leq & \int_{E_N} \Big( \left( 2d + 1 \right) \log N + \left| \log f(x) \right| \Big) d\mathbb{P}_X < \epsilon_N \Big( (2 d + 1 ) \log N + 1 \Big)
\end{eqnarray}

By choosing $\epsilon_N$ such that $\epsilon_N \log N \rightarrow 0$ as $N \rightarrow \infty$ (For example $\epsilon_N=1/N$), we will have the limit
$\lim_{N \rightarrow \infty} \int_E \left| \mathbb{E}[\zeta_1|X=x] - \log f(x) \right| d\mathbb{P}_X=0$.

For any $x \in \Omega_3 \setminus E_N$, since $P_X(x,0)=0$, we know that $\mathbb{P} \left( \rho_{k,1}=0 | X=x \right)=0$, so $\tilde{k}_1=k$ with probability 1. Conditioning on $\rho_{k,1}=r>0$, the term $ \left| \mathbb{E}[\zeta_1|X=x] - \log f(x) \right| $ can be decomposed as:
\begin{eqnarray}
& & \Big| \mathbb{E}[\zeta_1|X=x] - \log f(x) \Big| \\
& = & \left| \int_{r=0}^{\infty} \Big( \mathbb{E}[\zeta_1|X=x,\rho_{k,1}=r] - \log f(x) \Big) dF_{\rho_{k,1}}(r) \right| \\ 
& \leq & \left| \int_{r=0}^{\infty} \left( \log P_X(x,r) \prod_{l=1}^d \frac{P_{\gls{pa_xl_var}}\left( \gls{pa_xl_val} ,r \right)}{P_{\gls{pa_n_xl_var}}\left( \gls{pa_n_xl_val} ,r \right)} - \log f(x) \right) dF_{\rho_{k,1}}(r) \right| \label{eq:th1_term1} \\
& & + \left| \int_{r=0}^{\infty} \left( \psi(k) - \log N - \log P_X(x,r) \right) dF_{\rho_{k,1}}(r) \right| \label{eq:th1_term2} \\
& & + \sum_{l=1}^d \bigg| \int_{r=0}^{\infty} \Big( \mathbb{E} \left[ \log(n_{\gls{pa_n_xl_var}}^{(1)}+1)|(X,\rho_{k,1})=(x,r) \right] \nonumber\\ 
& & \hspace{.2in} - \log N P_{\gls{pa_n_xl_var}}(\gls{pa_n_xl_val},r) \Big) dF_{\rho_{k,1}}(r) \bigg| \label{eq:th1_term3}\\
& & + \sum_{l=1}^d \mathbf{1}_{ \{ \gls{pa_xl_var} \neq \emptyset \} } \bigg| \int_{r=0}^{\infty} \Big( \mathbb{E} \left[ \log(n_{\gls{pa_xl_var}}^{(1)}+1)|(X,\rho_{k,1})=(x,r) \right] \nonumber \\
& & \hspace{.2in} - \log N P_{\gls{pa_xl_var}}(\gls{pa_xl_val},r) \Big) dF_{\rho_{k,1}}(r) \bigg| \label{eq:th1_term4}
\end{eqnarray}

In which $F_{\rho_{k,1}}(r)$ is the CDF of the $k$-nearest neighbor distance $\rho_{k,1}$ given $X=x$. The derivative of this CDF with respect to $P_X(x,r)$ is given by:
\begin{equation} \frac{d F_{\rho_{k,1}}(r)}{d P_X(x,r)} = \frac{(N-1)!}{(k-1)!(N-k-1)!}P_X(x,r)^{k-1} \Big( 1- P_X(x,r) \Big)^{N-k-1}  \end{equation}
\\\\ 
\noindent\textit{Upper bound for the term (\ref{eq:th1_term1}) :}
\\ Since $P_X(x,r) \prod_{l=1}^d \frac{P_{\gls{pa_xl_var}}\left( \gls{pa_xl_val} ,r \right)}{P_{\gls{pa_n_xl_var}}\left( \gls{pa_n_xl_val} ,r \right)}$ converges to $f(x)$ uniformly over $\Omega_3 \setminus E$ as $P_X(x,r) \rightarrow 0$, So for any $\delta_N$  and any $x \in \Omega_3 \setminus E$, there exists $r_1$ such that for any $r < r_1$: 
\begin{equation} \left| \log P_X(x,r) \prod_{l=1}^d \frac{P_{\gls{pa_xl_var}}\left( \gls{pa_xl_val} ,r \right)}{P_{\gls{pa_n_xl_var}}\left( \gls{pa_n_xl_val} ,r \right)} - \log f(x) \right| < \delta_N \end{equation}
Let $r_2$ be the value of $r$ such that $P_X(x,r_2)=4 k \log N /N$, and take $r_N= \min \{ r_1, r_2 \}$. Thus $r_N$ depends on $x$, but $\delta_N$ does not depend on $x$ and $\lim_{N \rightarrow \infty} \delta_N = 0$. Therefore, (\ref{eq:th1_term1}) can be upper bounded as:

\begin{eqnarray}
& & \left| \int_{r=0}^{\infty} \left( \log P_X(x,r) \prod_{l=1}^d \frac{P_{\gls{pa_xl_var}}\left( \gls{pa_xl_val} ,r \right)}{P_{\gls{pa_n_xl_var}}\left( \gls{pa_n_xl_val} ,r \right)} - \log f(x) \right) dF_{\rho_{k,1}}(r) \right| \\
& \leq & \int_{r=0}^{r_N}  \left| \log P_X(x,r) \prod_{l=1}^d \frac{P_{\gls{pa_xl_var}}\left( \gls{pa_xl_val} ,r \right)}{P_{\gls{pa_n_xl_var}}\left( \gls{pa_n_xl_val} ,r \right)} - \log f(x) \right| dF_{\rho_{k,1}}(r) \nonumber \\
& & + \int_{r=r_N}^{\infty}  \left| \log P_X(x,r) \prod_{l=1}^d \frac{P_{\gls{pa_xl_var}}\left( \gls{pa_xl_val} ,r \right)}{P_{\gls{pa_n_xl_var}}\left( \gls{pa_n_xl_val} ,r \right)} - \log f(x) \right| dF_{\rho_{k,1}}(r) \\
& \leq & \delta_N \mathbb{P} \left( \rho_{k,1} \leq r_N | X=x \right) \nonumber \\
& & + \left(  \sup_{r \geq r_N} \left| \log P_X(x,r) \prod_{l=1}^d \frac{P_{\gls{pa_xl_var}}\left( \gls{pa_xl_val} ,r \right)}{P_{\gls{pa_n_xl_var}}\left( \gls{pa_n_xl_val} ,r \right)} - \log f(x) \right| \right) \mathbb{P} \left( \rho_{k,1} \geq r_N | X=x \right) \hspace{.3in}
\end{eqnarray}

First, $\mathbb{P} \left( \rho_{k,1} \leq r_N | X=x \right)$ is smaller than 1. Secondly, since $P_X(x,r) \geq 4k \log N / N > 1/N$ for $r\geq r_N$, so we have 
$\left| \log P_X(x,r) \right| \leq \log N$. The same bounds apply for any $\left| P_S(s,r) \right|$ as well, as $P_S(s,r) \geq P_X(x,r)$. Thus by triangle inequality, the supremum is upper-bounded by $(2d +1 ) \log N + \left| \log f(x) \right|$. Finally, the probability $\mathbb{P} \left( \rho_{k,1} \geq r_N | X=x \right)$ is upper bounded by:

\begin{eqnarray}
& & \mathbb{P} \left( \rho_{k,1} \geq r_N | X=x \right) \\
& = & \sum_{m=0}^{k-1} \left( \begin{array}{c} N-1 \\ m \end{array} \right) P_X(x,r_N)^m \left( 1 - P_X(x,r_N) \right)^{N-1-m} \\
& \leq & \sum_{m=0}^{k-1} N^m \left( 1 - P_X(x,r_N) \right)^{N-k} \\
& = & k N^k \left( 1-\frac{4k \log N}{N} \right)^{N/2} \\
& \leq & kN^k e^{-2k \log N} = \frac{k}{N^k}
\end{eqnarray}

for $N$ large enough such that $N - k > N/2$. Therefore \ref{eq:th1_term1} is upperbounded by:
\begin{eqnarray}
& & \left| \int_{r=0}^{\infty} \left( \log P_X(x,r) \prod_{l=1}^d \frac{P_{\gls{pa_xl_var}}\left( \gls{pa_xl_val} ,r \right)}{P_{\gls{pa_n_xl_var}}\left( \gls{pa_n_xl_val} ,r \right)} - \log f(x) \right) dF_{\rho_{k,1}}(r) \right| \\
& \leq & \delta_N + k \frac{(2 d +1 )\log N + \left| \log f(x) \right| }{N^k}
\end{eqnarray}
\\\\
\noindent\textit{Upper bound for the term (\ref{eq:th1_term2}) :}
\\We have:
\begin{eqnarray}
& & \int_{r=0}^{\infty} \left( \psi(k) - \log N - \log P_X(x,r)  \right) dF_{\rho_{k,1}}(r) \\
& = & \psi(k)-\log N - \frac{(N-1)!}{(k-1)!(N-k-1)!} \nonumber \\
& & \times \int_{r=0}^{\infty} \Big( P_X(x,r)^{k-1}(1-P_X(x,r))^{N-k-1}\log P_X(x,r) \Big) dP_X(x,r) \\
& = & \psi(k)-\log N - \frac{(N-1)!}{(k-1)!(N-k-1)!} \int_{t=0}^{1} \left( t^{k-1}(1-t)^{N-k-1}\log t \right) dt \\
& = & \psi(k)-\log N - \left( \psi(k) - \psi(N) \right) \\
& = & \psi(N) - \log(N)
\end{eqnarray}
We notice that for any $N$ we have $\psi(N) < \log N$ and $\lim_{N \rightarrow \infty} \left( \psi(N) - \log N \right)=0$.
\\\\
\noindent\textit{Upper bound for the individual terms of (\ref{eq:th1_term3}) and (\ref{eq:th1_term4}) :}
\\The upper bound for each of these terms follows the same logic, so we do the proof for an arbitrary subset $S \in \left\{ \gls{pa_xl_var} \right\}_{l=1}^d \cup \left\{ \gls{pa_n_xl_var} \right\}_{l=1}^d$. 

The distributions of $n_S^{(1)}$ for all $S \in \left\{ \gls{pa_xl_var} \right\}_{l=1}^d \cup \left\{ \gls{pa_n_xl_var} \right\}_{l=1}^d$ are given by the lemma below:

\begin{lemma}\label{lemma:n_dist}
Given $X=x$ and $\rho_{k,1}=r>0$ and $s$ being the projection of $x$ over $\mathcal{S}$, the term $n_S^{(1)}-k$ is distributed as $\text{Bino}\left( N-k-1, \frac{P_S(s,r)-P_X(x,r)}{1-P_X(x,r)} \right)$.
\end{lemma}
\begin{proof}
Please see the proof of Lemma B.2 in \cite{gao2017mixture}.
\end{proof}

The bound on $\mathbb{E}[ \log (n_S^{(1)}+1) | X=x,\rho_{k,1}=r]$ is given by the Lemma B.3 in \cite{gao2017mixture} which we restate here:

\begin{lemma}
If $X$ is distributed as $\text{Bino}(N, p)$, then $ \left| \mathbb{E} [ \log(X + k) ] - \log( N p + k) \right| \leq C/(N p + k)$ for some constant C.
\end{lemma}
\begin{proof}
Please see the proof of Lemma B.3 in \cite{gao2017mixture}.
\end{proof}

 Thus we can write:

\begin{eqnarray}
& & \left| \int_{r=0}^{\infty} \left( \mathbb{E} \left[ \log(n_S^{(1)}+1)|(X,\rho_{k,1})=(x,r) \right] - \log N P_{S}(s,r) \right) dF_{\rho_{k,1}}(r) \right| \\
& \leq & \left| \int_{r=0}^{\infty} \Bigg( \mathbb{E} [ \log(n_S^{(1)}+1)|(X,\rho_{k,1})=(x,r) ] \right. \nonumber \\ 
& & \left. - \log\left( (N-k-1)\frac{P_S(s,r)-P_X(x,r)}{1-P_X(x,r)}+k+1 \right) \Bigg) dF_{\rho_{k,1}}(r) \right|\\
& & + \left| \int_{r=0}^{\infty} \left( \log \frac{(N-k-1)\frac{P_S(s,r)-P_X(x,r)}{1-P_X(x,r)}+k+1}{N P_S(s,r)} \right) dF_{\rho_{k,1}}(r) \right|\\
& \leq & \int_{r=0}^{\infty} \left| \Bigg( \mathbb{E} \left[ \log(n_S^{(1)}+1)|(X,\rho_{k,1})=(x,r) \right] \right. \nonumber \\ 
& & \left. - \log\left( (N-k-1)\frac{P_S(s,r)-P_X(x,r)}{1-P_X(x,r)}+k+1 \right) \Bigg) \right| dF_{\rho_{k,1}}(r) \label{eq:th1_terms456_1}\\
& & + \left| \mathbb{E}_r \left[ \log \frac{N(P_S(s,r)-P_X(x,r)) + (k+1)(1-P_S(s,r))}{N P_S(s,r)(1-P_X(x,r))} \right] \right| \label{eq:th1_terms456_2}
\end{eqnarray}

Where $\mathbb{E}_r$ denotes the expectation over the distribution $F_{\rho_{k,1}}$. By the Lemma B.3  in \cite{gao2017mixture}, the term in \ref{eq:th1_terms456_1} is upper bounded by:
\begin{eqnarray}
& & \int_{r=0}^{\infty} \Bigg| \Bigg( \mathbb{E} \left[ \log(n_S^{(1)}+1)|(X,\rho_{k,1})=(x,r) \right] \nonumber \\ 
& & - \log \left( (N-k-1)\frac{P_S(s,r)-P_X(x,r)}{1-P_X(x,r)}+k+1 \right) \Bigg) \Bigg| dF_{\rho_{k,1}}(r)\\
& \leq & \int_{r=0}^{\infty} \frac{C}{(N-k-1)\frac{P_S(s,r)-P_X(x,r)}{1-P_X(x,r)}+k+1} dF_{\rho_{k,1}}(r) \\
& \leq & \int_{r=0}^{\infty} \frac{C}{k+1} dF_{\rho_{k,1}}(r) = \frac{C}{k+1}
\end{eqnarray}

The last inequality follows from the fact that $P_S(s,r) > P_X(x,r)$ . For (\ref{eq:th1_terms456_2}), using the fact that $\log (x/y) \leq (x-y)/y$ and Cauchy-Schwarz inequality, we have the following:
\begin{eqnarray}
& & \mathbb{E}_r \left[ \log \frac{N(P_S(s,r)-P_X(x,r)) + (k+1)(1-P_S(s,r))}{N P_S(s,r)(1-P_X(x,r))} \right] \\ 
& \leq & \mathbb{E}_r \left[ \frac{N(P_S(s,r)-P_X(x,r)) + (k+1)(1-P_S(s,r))}{N P_S(s,r)(1-P_X(x,r))} -1 \right] \\
& = & \mathbb{E}_r \left[ \frac{(k+1-N P_X(x,r))(1-P_S(s,r))}{N P_S(s,r)(1-P_X(x,r))} \right] \\
& \leq & \sqrt{ \mathbb{E}_r \left[ \left( \frac{k+1-N P_X(x,r)}{N P_X(x,r)} \right)^2 \right] 
\mathbb{E}_r \left[ \left( \frac{P_X(x,r)(1-P_S(s,r))}{ P_S(s,r)(1-P_X(x,r))} \right)^2 \right]}
\end{eqnarray}

Note that $P_S(s,r) \geq P_X(x,r)$ for all $r$, so the second expectation term is always smaller than or equal to 1. For the first expectation, let $t=P_X(x,r)$ then we have:  
\begin{eqnarray}
& & \mathbb{E}_r \left[ \left( \frac{k+1-N P_X(x,r)}{N P_X(x,r)} \right)^2 \right] \\
& = & \int_{r=0}^{\infty} \left( \frac{k+1-N P_X(x,r)}{N P_X(x,r)} \right)^2 dF_{\rho_{k,1}}(r) \\
& = & \frac{(N-1)!}{(k-1)!(N-k-1)!} \int_{t=0}^1 \frac{(k+1-Nt)^2}{N^2t^2} t^{k-1} (1-t)^{N-k-1} dt \\
& = & \frac{(N-1)(N-2)(k+1)^2}{N^2(k-1)(k-2)}-\frac{2(N-1)(k+1)}{N(k-1)}+1
\end{eqnarray}

This term is upper bounded by $C_1(1/N+1/k)$ for some constant $C_1$ for $N$ and $k$ large enough. Thus:
\begin{equation} \mathbb{E}_r \left[ \log \frac{N(P_S(s,r)-P_X(x,r)) + (k+1)(1-P_S(s,r))}{N P_S(s,r)(1-P_X(x,r))} \right] \leq \sqrt{C_1(\frac{1}{N}+\frac{1}{k})} \end{equation}

Similarly, by using the fact that $\log(x/y)>(x-y)/x$ and Cauchy-Schwarz inequality again, we conclude that there are some constant $C_2>0$ such that:
\begin{equation} \mathbb{E}_r \left[ \log \frac{N(P_S(s,r)-P_X(x,r)) + (k+1)(1-P_S(s,r))}{N P_S(s,r)(1-P_X(x,r))} \right] \geq -\sqrt{C_2(\frac{1}{N}+\frac{1}{k})} \end{equation}

Therefore by combining all these bounds we obtain
\begin{eqnarray}
& & \left| \int_{r=0}^{\infty} \bigg( \mathbb{E} \left[ \log(n_S^{(1)}+1)|(X,\rho_{k,1})=(x,r) \right] - \log N P_{S}(s,r) \bigg) dF_{\rho_{k,1}}(r) \right| \\
& \leq & \frac{C}{k+1} + \sqrt{C'(\frac{1}{k}+\frac{1}{N})}
\end{eqnarray}
where $C'=\max\{ C_1,C_2\}$.

\vspace{.3in}
\noindent Now putting all the bounds for $\Omega_3 \setminus E_N$, we have:
\begin{eqnarray}
& & \int_{\Omega_3 \setminus E_N} \bigg| \mathbb{E} \left[ \zeta_1|X=x \right] - \log f(x) \bigg| d\mathbb{P}_X \\
& \leq & \delta_N + k\frac{(2 d + 1 ) \log N +  \int_{\mathcal{X}} \left| \log f(x) \right| d\mathbb{P}_X}{N^k} \nonumber \\
& & + \psi(N) - \log(N) + 2 d \left( \frac{C}{k+1} + \sqrt{C'(\frac{1}{k}+\frac{1}{N})} \right)   
\end{eqnarray}
From the assumptions \ref{assumptions}, $k$ increases as $N \rightarrow \infty$, and $ \int_{\mathcal{X}} \left| \log f(x) \right| d\mathbb{P}_X < \infty$. Therefore the whole upperbound vanishes as $N$ goes to infinity. Thus for the entire set $\Omega_3$ we have:
\begin{equation} \lim_{N \rightarrow \infty} \int_{\Omega_3} \bigg| \mathbb{E}[\zeta_1|X=x)] - \log f(x) \bigg| d\mathbb{P}_X = 0  \end{equation}


\subsection{Proof of Lemma \ref{lemma:RN_convergence}} \label{sec:lemma_RN_convergence_proof}

This proof is based on the Lebesgue-Besicovitch differentiation theorem (For example look at Theorem 1.32 from \cite{evans2018measure}), stated as:
\begin{theorem} \label{thm:leb-bes}
Let $\mu$ be a Radon measure on $\mathbb{R}^d$. For $f \in L_{loc}^1(\mu)$,
\begin{equation} \lim_{r \rightarrow 0} \frac{1}{\mu \left( B_r(x) \right)} \int_{B_r(x)} f d\mu = f(x)\end{equation}
for $\mu$-a.e. $x$.
\end{theorem}

Let $f = \frac{d \mathbb{P}_X}{d \gls{graph_dist}_X}$ and $ \mu = \gls{graph_dist}_X$. Since $\mu$ is a probability measure, it is a Radon measure on Euclidean space. Also, since $\int_{\mathcal{X}} |f| d\mu = 1$, so it's globally and therefore locally integrable with respect to $\mu$. Thus the conditions of the Theorem \ref{thm:leb-bes} are satisfied, and we can write:

\begin{eqnarray}
& & \lim_{r\rightarrow 0} P_X(x,r) \prod_{l=1}^d \frac{P_{\gls{pa_xl_var}}\left( \gls{pa_xl_val} ,r \right)}{P_{\gls{pa_n_xl_var}}\left(  \gls{pa_n_xl_val} ,r \right)}\\
& = & \lim_{r\rightarrow 0} \frac{\mathbb{P} \Big\{ B_r(x) \Big\} }{\prod_{l=1}^d \mathbb{P}_{ X_l | \gls{pa_xl_var} } \Big\{ B_r \left( x_l \right) \Big| B_r \left( \gls{pa_xl_val} \right) \Big\}}\\
& = & \lim_{r\rightarrow 0} \frac{\mathbb{P} \Big\{ B_r(x) \Big\}}{\gls{graph_dist}_X \Big\{ B_r(x) \Big\}}\\
& = & \lim_{r\rightarrow 0} \frac{1}{\gls{graph_dist}_X \Big\{ B_r(x) \Big\}} \int_{B_r(x)} \frac{d \mathbb{P}_X}{d \gls{graph_dist}_X} d \gls{graph_dist}_X \\ 
& = & \frac{d \mathbb{P}_X}{d \gls{graph_dist}_X}(x) 
\end{eqnarray}


\subsection{Proof of Lemma \ref{lemma:omega2_upperbound}} \label{sec:lemma_omega2_upperbound_proof}

First, we upperbound $\mathbb{E} \left[ \log(X+m) \middle| X \geq k \right] - \log(Np)$. We can see that:
\begin{eqnarray}
& & \mathbb{E} \left[ X+m \middle| X \geq k \right] \\
& = & \frac{1}{ \mathbb{P} \left( X \geq k \right)} \sum_{i=k}^N (i+m) \left( \begin{array}{c} N \\ i \end{array}  \right) p^i (1-p)^{N-i} \\
& \leq & \frac{1}{1 - \exp\left( -2 \frac{(Np-k)^2}{N} \right)} \sum_{i=k}^N (i+m) \left( \begin{array}{c} N \\ i \end{array}  \right) p^i (1-p)^{N-i} \\
& \leq & \frac{1}{1 - \exp\left( -2 \frac{(Np-k)^2}{N} \right)} \sum_{i=1}^N (i+m) \left( \begin{array}{c} N \\ i \end{array}  \right) p^i (1-p)^{N-i} \\
& = & \frac{1}{1 - \exp\left( -2 \frac{(Np-k)^2}{N} \right)} \left( \mathbb{E}\left[ X \right] + m \right) = \frac{Np+m}{1 - \exp\left( -2 \frac{(Np-k)^2}{N} \right)}
\end{eqnarray}

In which we used the Hoeffding's inequality. We know that $\mathbb{E} \left[ \log(X+m) \middle| X \geq k \right] \leq \log \left( \mathbb{E} \left[ X+m \middle| X \geq k \right] \right)$, therefore: 
\begin{eqnarray}
\mathbb{E} \left[ \log(X) \middle| X \geq k \right] - \log(Np) \leq \log \left( \frac{1+\frac{m}{Np}}{1 - \exp\left( -2 \frac{(Np-k)^2}{N} \right)} \right)
\end{eqnarray}

Second, to give an upper bound over $\log(Np) - \mathbb{E} \left[ \log(X+m) \middle| X \geq k \right]$, we first notice that:
\begin{equation} \log(Np) - \mathbb{E} \left[ \log(X+m) \middle| X \geq k \right] \leq \log(Np) - \mathbb{E} \left[ \log(X) \middle| X \geq k \right] \end{equation} 

Then we upperbound $\log(Np) - \mathbb{E} \left[ \log(X) \middle| X \geq k \right]$ by applying Taylor's theorem around $x_0=Np$, where there exists $\zeta$ between $x$ and $x_0$ such that:
\begin{equation} \log(x) = \log(Np) + \frac{x-Np}{Np} - \frac{(x-Np)^2}{2\zeta^2} \end{equation}
since $\zeta \geq \min \left\{ x, x_0 \right\} = \min \left\{ x, Np \right\}$, we have:

\begin{eqnarray}
& & -\log(x) + \log(Np) + \frac{x-Np}{Np} = \frac{(x-Np)^2}{2 \zeta^2} \nonumber \\
& \leq & \max \left\{ \frac{(x-Np)^2}{2x^2}, \frac{(x-NP)^2}{2(Np)^2} \right\} \leq \frac{(x-Np)^2}{2x^2} + \frac{(x-Np)^2}{2(Np)^2 }
\end{eqnarray}

Now taking the conditional expectations from both sides, we have:

\begin{eqnarray}
 & & -\mathbb{E} \left[ \log(X) \middle| X \geq k \right] + \log(Np) +\frac{  \mathbb{E} \left[ X \middle| X \geq k \right] -Np}{Np} \nonumber \\
 & \leq & \mathbb{E} \left[ \frac{(X-Np)^2}{2X^2}  \middle| X \geq k \right] + \frac{ \mathbb{E} \left[ (X-Np)^2  \middle| X \geq k \right]}{2(Np)^2 }
\end{eqnarray}

First, we notice that $ \mathbb{E} \left[ X \middle| X \geq k \right] \geq  \mathbb{E} \left[ X \right] = Np$.

Second, $\mathbb{E} \left[ (X-Np)^2  \middle| X \geq k \right] \leq \frac{1}{1 - \exp\left( -2 \frac{(Np-k)^2}{N} \right)} \text{Var} \left[ X \right] = \frac{Np(1-p)}{1 - \exp\left( -2 \frac{(Np-k)^2}{N} \right)}$.

Thus we can write:
\begin{eqnarray}
-\mathbb{E} \left[ \log(X) \middle| X \geq k \right] + \log(Np) \leq \frac{Np(1-p)}{1 - \exp\left( -2 \frac{(Np-k)^2}{N} \right)}\frac{1}{2(Np)^2} + \mathbb{E} \left[ \frac{(X-Np)^2}{2X^2}  \middle| X \geq k \right] \label{eq:lemma_b3_modified}
\end{eqnarray}

To deal with the term $\mathbb{E} \left[ \frac{(X-Np)^2}{2X^2}  \middle| X \geq k \right]$, we have:

\begin{eqnarray}
& & \mathbb{E} \left[ \frac{(X-Np)^2}{2X^2}  \middle| X \geq k \right] \\
&\leq& \frac{1}{1 - \exp\left( -2 \frac{(Np-k)^2}{N} \right)} 
\sum_{i=k}^N \frac{(i-Np)^2}{2 i^2} \left( \begin{array}{c} N \\ i \end{array} \right) p^i (1-p)^{N-i} \\
&\leq& \frac{1}{1 - \exp\left( -2 \frac{(Np-k)^2}{N} \right)} 
\sum_{i=k}^N \frac{(i-Np)^2}{(i+1)(i+2)} \left( \begin{array}{c} N \\ i \end{array} \right) p^i (1-p)^{N-i} \\
&=& \frac{1}{1 - \exp\left( -2 \frac{(Np-k)^2}{N} \right)} 
\sum_{i=k}^N \frac{(i-Np)^2}{(N+1)(N+2)p^2} \left( \begin{array}{c} N+2 \\ i+2 \end{array} \right) p^{2+i} (1-p)^{N-i} \\
&\leq& \frac{1}{1 - \exp\left( -2 \frac{(Np-k)^2}{N} \right)}\frac{1}{(N+1)(N+2)p^2}
\mathbb{E}_{Y \sim \text{Bino}(N+2,p)} \left[ (Y-Np)^2 \right] \\
& = & \frac{1}{1 - \exp\left( -2 \frac{(Np-k)^2}{N} \right)}\frac{(N+2)p(1-p)+4p^2}{(N+1)(N+2)p^2} \\
& \leq & \frac{1}{1 - \exp\left( -2 \frac{(Np-k)^2}{N} \right)}\frac{(N+2)p}{(N+1)(N+2)p^2} \leq \frac{1}{1 - \exp\left( -2 \frac{(Np-k)^2}{N} \right)}\frac{1}{Np} 
\end{eqnarray}

In which we used the fact that $ 2i^2 \geq (i+1)(i+2)$ for $i \geq 4$, and $(N+2)p \geq 4p$ for $N \geq 2$. Plugging it into Equation \ref{eq:lemma_b3_modified}, we have:

\begin{equation}
-\mathbb{E} \left[ \log(X) \middle| X \geq k \right] + \log(Np) \leq \frac{1}{1 - \exp\left( -2 \frac{(Np-k)^2}{N} \right)}\frac{3}{2Np}
\end{equation}

And the desired result is yielded.

\subsection{Proof of Theorem \ref{thm:variance_convergence}} \label{sec:variance_convergence_proof}

For this part, we also follow the same procedure as followed for the Theorem 2 in \cite{gao2017mixture} using Efron-Stein inequality. Suppose $\widehat{\gls{graph_distance_measure}}^{(N)}(X,\mathcal{G})$ is the estimate based on the original samples $ x^{(1)}, x^{(2)}, \ldots, x^{(N)} $. For the usage of Efron-Stein inequality, we suppose that there is another set of $n$ i.i.d samples drawn from $P_X$ denoted by $ x'^{(1)}, x'^{(2)}, \ldots, x'^{(n)}$. Let $\widehat{\gls{graph_distance_measure}}^{(N)}(X^{(j)},\mathcal{G})$ be the estimate based on the original samples' set in which the $j$th sample is replaced by another i.i.d sample $x'^{(j)}$ taken from the second set, i.e. $\widehat{\gls{graph_distance_measure}}^{(N)}(X^{(j)},\mathcal{G})$ is the estimate based on $x^{(1)}, \ldots, x^{(j-1)}, x'^{(j)}, x^{(j+1)}, \ldots, x^{(N)} $. Then the Efron-Stein inequality states that:
\begin{equation} \text{Var}\left[ \widehat{\gls{graph_distance_measure}}^{(N)}(X,\mathcal{G}) \right] \leq \frac{1}{2} \sum_{j=1}^N \mathbb{E} \left[ \left(  \widehat{\gls{graph_distance_measure}}^{(N)}(X,\mathcal{G}) -  \widehat{\gls{graph_distance_measure}}^{(N)}(X^{(j)},\mathcal{G}) \right)^2 \right] \end{equation}

Now we will find an upper bound for the difference $\left|  \widehat{\gls{graph_distance_measure}}^{(N)}(X,\mathcal{G}) -  \widehat{\gls{graph_distance_measure}}^{(N)}(X^{(j)},\mathcal{G}) \right|$ for a given index $j$ by considering the worst case scenario. Let  $\widehat{\gls{graph_distance_measure}}^{(N)}(X_{\setminus j},\mathcal{G})$ be the estimate based on the original samples' set in which the $j$th sample is removed, i.e. $x^{(1)}, \ldots, x^{(j-1)}, x^{j+1}, \ldots, x^{(N)}$. Then by triangle inequality, we have:
\begin{eqnarray}
& & \sup_{x^{(1)}, \ldots, x^{(N)}, x'^{(j)}} \left|  \widehat{\gls{graph_distance_measure}}^{(N)}(X,\mathcal{G}) -  \widehat{\gls{graph_distance_measure}}^{(N)}(X^{(j)},\mathcal{G}) \right| \nonumber\\
& \leq & \sup_{x^{(1)}, \ldots, x^{(N)}, x'^{(j)}} \Bigg( \left| \widehat{\gls{graph_distance_measure}}^{(N)}(X,\mathcal{G}) - \widehat{\gls{graph_distance_measure}}^{(N)}(X_{\setminus j},\mathcal{G}) \right| \\
& & + \left| \widehat{\gls{graph_distance_measure}}^{(N)}(X_{\setminus j},\mathcal{G}) -  \widehat{\gls{graph_distance_measure}}^{(N)}(X^{(j)},\mathcal{G}) \right|  \Bigg) \nonumber\\
& \leq & \sup_{x^{(1)}, \ldots, x^{(N)}}\left| \widehat{\gls{graph_distance_measure}}^{(N)}(X,\mathcal{G}) - \widehat{\gls{graph_distance_measure}}^{(N)}(X_{\setminus j},\mathcal{G}) \right| \\
& & + \sup_{x^{(1)}, \ldots, x^{(j-1)}, x'^{(j)}, x^{(j+1)}, \ldots, x^{(N)}} \left| \widehat{\gls{graph_distance_measure}}^{(N)}(X_{\setminus j},\mathcal{G}) -  \widehat{\gls{graph_distance_measure}}^{(N)}(X^{(j)},\mathcal{G}) \right| \nonumber\\
& = & 2 \sup_{x^{(1)}, \ldots, x^{(N)}} \left| \widehat{\gls{graph_distance_measure}}^{(N)}(X,\mathcal{G}) - \widehat{\gls{graph_distance_measure}}^{(N)}(X_{\setminus j},\mathcal{G}) \right|
\end{eqnarray}

Remember that: 
\begin{eqnarray} 
& & \widehat{\gls{graph_distance_measure}}^{(N)}(X,\mathcal{G}) \nonumber \\
& = & \frac{1}{N} \sum_{i=1}^N \zeta_i(X) \\
& = & \frac{1}{N} \sum_{i=1}^N \Bigg( \psi(\tilde{k}_i) + \sum_{l=1}^d \left(  \mathbf{1}_{ \{ \gls{pa_xl_var} \neq \emptyset \} } \log (n_{\gls{pa_xl_var}}^{(i)}+1) - \log (n_{\gls{pa_n_xl_var}}^{(i)}+1) \right) \nonumber\\
& &+ \bigg( \sum_{l=1}^d \mathbf{1}_{ \{  \gls{pa_xl_var}= \emptyset \} } -1 \bigg) \log N \Bigg) 
\end{eqnarray} 

Thus we can write:
\begin{equation} \sup_{x^{(1)}, \ldots, x^{(N)}}\left| \widehat{\gls{graph_distance_measure}}^{(N)}(X,\mathcal{G}) - \widehat{\gls{graph_distance_measure}}^{(N)}(X^{(j)},\mathcal{G}) \right|
\leq \frac{2}{N} \sup_{x^{(1)}, \ldots, x^{(N)}} \sum_{i=1}^N \left| \zeta_i(X) - \zeta_i(X_{\setminus j}) \right| \end{equation}

\noindent Now we need to upper-bound $\left| \zeta_i(X) - \zeta_i(X_{\setminus j}) \right|$ for all the different $i$'s. The cases are as follows:

\vspace{.2in}
\noindent \textbf{Case I: \boldmath$i=j$ }. Since the upper bounds $\left| \zeta_i(X) \right| \leq (2 d + 1) \log N$ and $\left| \zeta_i(X_{\setminus j}) \right| \leq ( 2 d + 1 ) \log (N-1)$ always holds, thus we have $\left| \zeta_i(X) - \zeta_i(X_{\setminus j}) \right| \leq 2 ( 2 d + 1 ) \log N$. Since there's only one $i$ equal to $j$, we have $\sum_{\text{Case I}} \left| \zeta_i(X) - \zeta_i(X_{\setminus j}) \right| \leq 2 ( 2 d + 1 ) \log N$.

\vspace{.2in}
\noindent \textbf{Case II: \boldmath$i\neq j$ and \boldmath$\rho_{k,i}=0$ }. Suppose $S \in \big\{ \gls{pa_xl_var} \big\}_{l=1}^d \cup \big\{ \gls{pa_n_xl_var} \big\}_{l=1}^d \cup \left\{ X \right\}$ is a subset of $X$. Since $\rho_{k,i}=0$,  then $n_S^{(i)}=\left| \{ i' \neq i : s^{(i)} = s^{(i')} \} \right|$. We recall that for $S=X$ we actually have $n_S^{(i)}=\tilde{k}_i$. Thus by removing the point $x^{(j)}$ is this case, for any subset $S$, if $s^{(i)}=s^{(j)}$, then $n_S^{(i)}$ is decreased by 1, and if $s^{(i)} \neq s^{(j)}$ then $n_S^{(i)}$ will not change. Thus we can write:
\begin{eqnarray}
& & \left| \zeta_i(X) - \zeta_i(X_{\setminus j}) \right| \\
& \leq & \left| \psi \left( \tilde{k}_i \right) - \psi \left( \tilde{k}_i-1 \right) \right| \\
& & + \sum_{l=1}^d \mathbf{1}_{ \{ \gls{pa_xl_var} \neq \emptyset \} }  \mathbf{1}_{ \left\{ {\gls{pa_xl_val}}^{(i)} = {\gls{pa_xl_val}}^{(j)} \right\} } \bigg| \log (n_{\gls{pa_xl_var}}^{(i)}+1) - \log n_{\gls{pa_xl_var}}^{(i)} \bigg| \\
& & + \sum_{l=1}^d \mathbf{1}_{ \left\{ {\gls{pa_n_xl_val}}^{(i)} = {\gls{pa_n_xl_val}}^{(j)} \right\} } \bigg| \log (n_{\gls{pa_n_xl_var}}^{(i)}+1) - \log n_{\gls{pa_n_xl_var}}^{(i)} \bigg| \\
& & + \bigg( \sum_{l=1}^d \mathbf{1}_{ \left\{  \gls{pa_xl_var}= \emptyset \right\} } +1 \bigg) \Big( \log N - \log(N-1) \Big) \\
& \leq & \frac{1}{\tilde{k}_i -1} + \sum_{l=1}^d \mathbf{1}_{ \left\{ \gls{pa_xl_var} \neq \emptyset \right\} }  \mathbf{1}_{ \left\{ {\gls{pa_xl_val}}^{(i)} = {\gls{pa_xl_val}}^{(j)} \right\} } \frac{1}{n_{\gls{pa_xl_var}}^{(i)}}\\
& & + \sum_{l=1}^d \mathbf{1}_{ \left\{ {\gls{pa_n_xl_val}}^{(i)} = {\gls{pa_n_xl_val}}^{(j)} \right\} }\frac{1}{n_{\gls{pa_n_xl_var}}^{(i)}} \\
& & + \bigg( \sum_{l=1}^d \mathbf{1}_{ \left\{  \gls{pa_xl_var}= \emptyset \right\} } +1 \bigg) \frac{1}{N-1}
\end{eqnarray}
Where we used the fact that $\log(1+1/x) \leq 1/x$. For any $S$, the number of nodes $i$ satisfying $s^{(i)} = s^{(j)}$ is no more than $n_S^{(j)}$, and the total number of points in the case II is no more than $N-1$. Thus we can write:
\begin{eqnarray} 
\sum_{i \in \text{Case II}} \left| \zeta_i(X) - \zeta_i(X_{\setminus j}) \right| & \leq & (\tilde{k}_i -1) \frac{1}{\tilde{k}_i -1} \\ 
& & + \sum_{l=1}^d \mathbf{1}_{ \left\{ \gls{pa_xl_var} \neq \emptyset \right\} } n_{\gls{pa_xl_var}}^{(i)} \frac{1}{n_{\gls{pa_xl_var}}^{(i)}}\\
& & + \sum_{l=1}^d n_{\gls{pa_n_xl_var}}^{(i)}\frac{1}{n_{\gls{pa_n_xl_var}}^{(i)}} \\
& & + \bigg( \sum_{l=1}^d \mathbf{1}_{ \left\{  \gls{pa_xl_var}= \emptyset \right\} } +1 \bigg) (N-1)\frac{1}{N-1} \\
& = & 2 (d+1)
\end{eqnarray}

\vspace{.2in}
\noindent \textbf{Case III: \boldmath$i\neq j$ and \boldmath$\rho_{k,i}>0$ }. In this case, $\tilde{k}_i$ is always equal to $k$, and for any subset $S \in \big\{ \gls{pa_xl_var} \big\}_{l=1}^d \cup \big\{ \gls{pa_n_xl_var} \big\}_{l=1}^d \cup \left\{ X \right\}$ we have $n_S^{(i)}= \left| \{ i' \neq i : \| s^{(i)} - s^{(j)} \| \leq \rho_{k,i} \} \right|$.

\vspace{.2in}
\noindent \textit{Case III.1: $\| x^{(i)} - x^{(j)} \| \leq \rho_{k,i}$. } This means that the point $x_j$ is in the $k$ nearest neighbors of $x_i$ and by eliminating it, $\rho_{k,i}$ will change. So we don't know how $n_S^{(i)}$'s will change, thus we will use the upper-bound $ \left| \zeta_i(X) - \zeta_i(X_{\setminus j}) \right| \leq 2 ( 2 d + 1 ) \log N$. From the first part of the Lemma C.1 from \cite{gao2017mixture}, we can upper bound the number of such $i$'s by $\gamma_d k$, in which $\gamma_d$ is the minimum number of $d$-dimensional cones with the angle smaller than $\pi/6$ needed to cover the total space $\mathbb{R}^d$. Thus we have:
\begin{equation} \sum_{\text{Case III.1}} \left| \zeta_i(X) - \zeta_i(X_{\setminus j}) \right| \leq 2 k ( 2 d + 1 ) \gamma_{d_X} \log N  \end{equation}
where $d_X$ is the dimension of the space $\mathcal{X}$.

\vspace{.2in}
\noindent \textit{Case III.2: $\| x^{(i)} - x^{(j)} \| > \rho_{k,i}$. } This case is somewhat similar to the Case II. Here the  $\rho_{k,i}$ will not change. But the $n_S^{(i)}$'s for all subsets $S \neq X$ can decrease by 1. We can write:
\begin{eqnarray}
& & \sum_{\text{Case III.2}} \left| \zeta_i(X) - \zeta_i(X_{\setminus j}) \right| \\
& \leq & \sum_{i=1}^N \sum_{l=1}^d \mathbf{1}_{ \left\{ \gls{pa_xl_var} \neq \emptyset \right\} }  \mathbf{1}_{ \left\{ \gls{pa_xl_val}^{(i)} = {\gls{pa_xl_val}}^{(j)} \right\}} \bigg| \log (n_{\gls{pa_xl_var}}^{(i)}+1) - \log n_{\gls{pa_xl_var}}^{(i)} \bigg| \\
& & + \sum_{i=1}^N \sum_{l=1}^d \mathbf{1}_{ \left\{ {\gls{pa_n_xl_val}}^{(i)} = {\gls{pa_n_xl_val}}^{(j)} \right\} } \bigg| \log (n_{\gls{pa_n_xl_var}}^{(i)}+1) - \log n_{\gls{pa_n_xl_var}}^{(i)} \bigg| \\
& & + \sum_{i=1}^{N-1} \bigg( \sum_{l=1}^d \mathbf{1}_{ \left\{  \gls{pa_xl_var}= \emptyset \right\} } +1 \bigg) \Big( \log N - \log(N-1) \Big) \\
& \leq & \sum_{i=1}^N\sum_{l=1}^d \mathbf{1}_{ \left\{ \gls{pa_xl_var} \neq \emptyset \right\} }  \mathbf{1}_{ \left\{ {\gls{pa_xl_val}}^{(i)} = {\gls{pa_xl_val}}^{(j)} \right\} } \frac{1}{n_{\gls{pa_xl_var}}^{(i)}}\\
& & + \sum_{i=1}^N\sum_{l=1}^d \mathbf{1}_{ \left\{ {\gls{pa_n_xl_val}}^{(i)} = {\gls{pa_n_xl_val}}^{(j)} \right\} }\frac{1}{n_{\gls{pa_n_xl_var}}^{(i)}} \\
& & + \sum_{i=1}^{N-1}\bigg( \sum_{l=1}^d \mathbf{1}_{ \left\{  \gls{pa_xl_var}= \emptyset \right\} } +1 \bigg) \frac{1}{N-1}\\
& \leq & \sum_{l=1}^d \mathbf{1}_{ \left\{ \gls{pa_xl_var} \neq \emptyset \right\} } \gamma_{d_{\gls{pa_xl_var}}} \log (N+1) \\
& & + \sum_{l=1}^d \gamma_{d_{\gls{pa_n_xl_var}}} \log (N+1) + \sum_{l=1}^d \mathbf{1}_{ \left\{  \gls{pa_xl_var}= \emptyset \right\} } +1 \label{eq:lemmab2_orig_paper} \\
& \leq & 2 d \gamma_{d_X} \log (N+1) + 1 
\end{eqnarray}

The inequality \ref{eq:lemmab2_orig_paper} follows from the second part of the Lemma C.1 from \cite{gao2017mixture}. 

\vspace{.3in}
Now combining all three cases together, we have:
\begin{eqnarray}
& & \sum_{i=1}^N \left| \zeta_i(X) - \zeta_i(X_{\setminus j}) \right| \nonumber \\
& \leq & 2 ( 2 d + 1 ) \log N + 2(d + 1) + 2 k ( 2 d + 1 ) \gamma_{d_X} \log N  + 2 d \gamma_{d_X} \log (N+1) + 1 \nonumber \\
& \leq & 6 k ( 2 d + 1 ) \gamma_{d_X}  \log (N+1)
\end{eqnarray}

Thus:
\begin{equation} \sup_{x^{(1)}, \ldots, x^{(N)}, x'^{(j)}} \left| \widehat{\gls{graph_distance_measure}}^{(N)}(X,\mathcal{G}) - \widehat{\gls{graph_distance_measure}}^{(N)}(X^{(j)},\mathcal{G}) \right| \leq \frac{12 k ( 2 d + 1 ) \gamma_{d_X}  \log (N+1)}{N} \end{equation}

And:

\begin{eqnarray}
& & \text{Var}\left[ \widehat{\gls{graph_distance_measure}}^{(N)}(X,\mathcal{G}) \right] \\
& \leq & \frac{1}{2} \sum_{j=1}^N \mathbb{E} \left[ \left(  \widehat{\gls{graph_distance_measure}}^{(N)}(X,\mathcal{G}) - \widehat{\gls{graph_distance_measure}}^{(N)}(X^{(j)},\mathcal{G}) \right)^2 \right] \\
& \leq & \frac{1}{2} \sum_{j=1}^N \sup_{x^{(1)}, \ldots, x^{(N)}, x'^{(j)}} \left(  \widehat{\gls{graph_distance_measure}}^{(N)}(X,\mathcal{G}) - \widehat{\gls{graph_distance_measure}}^{(N)}(X^{(j)},\mathcal{G}) \right)^2 \\
& \leq & \frac{1}{2} \sum_{j=1}^N \left( \frac{12 k ( 2 d + 1 ) \gamma_{d_X}  \log (N+1)}{N} \right)^2 \\
& = & \frac{72 \gamma_{d_X} ^2 \bigg( k ( 2 d + 1 ) \log N \bigg)^2}{N}
\end{eqnarray}

Therefore if $( k_N \log N )^2 / N \rightarrow 0 $ as $N$ goes to infinity, we have $\lim_{N \rightarrow \infty} \text{Var} \left[ \widehat{\gls{graph_distance_measure}}^{(N)}(X,\mathcal{G}) \right]=0$.


\section{Details of Numerical Experiments}\label{sec:exp_details}

In this section, we will discuss the carried out numerical experiments done in more details. This includes but not limited to the choice of parameters, the data generated, and the derivation of the theory values.  The parameter of the nearest neighbor $k$ for all the mixture, legacy and $\Sigma H$ methods is set to $\sqrt{N}/5$ to conform with the Assumptions \ref{assumptions} and keep the computational complexity at an acceptable level. The number of bins in binning method at each dimension is kept at $\sqrt[d]{N/m}$ for all the algorithms so we roughly have $m$ samples per each bin, and we choose $m \in \{ 10,20,100 \}$ whichever giving the best precision.

\subsection{Experiment 1: Markov chain model with continuous-discrete mixture}

For the first experiment, we simulated an $X$-$Z$-$Y$ Markov chain model in which the random variable $X$ is chosen as $X = \min \Big( \alpha_1, \tilde{X} \Big)$ where $\tilde{X} \sim \mathcal{U}(0,1)$ represents an auxiliary random variable uniformly distributed between $0$ and $1$. This means that we first generate a sample from $\tilde{X}$ denoted by $\tilde{x}$  and then let $x = \min \{ \alpha_1, \tilde{x} \}$. Subsequently:
\begin{eqnarray}
Z & = & \min \left( X, \alpha_2 \right) \\
Y & = & \min \left( Z, \alpha_3 \right)
\end{eqnarray}
 We assume that $\alpha_3 < \alpha_2 < \alpha_1$. The three variables $X$, $Y$ and $Z$ represent a mixture of continuous and discrete random variables.

We simulated this system for various numbers of samples while setting $\alpha_1=0.9$, $\alpha_2=0.8$ and $\alpha_3=0.7$. For each set of samples $I(X;Y|Z)$ is estimated via different methods and its theory value is obviously equal to $0$. The results are shown in Figure \ref{fig:exp1_cmi_vs_n}.

\subsection{Experiment 2: Mixture of AWGN and BSC channels with variable error probability}

As the second scheme of our experiments, we considered an Additive White Gaussian Noise (AWGN) Channel in parallel with a Binary Symmetric Channel (BSC) where only one of them can be activated at a time. The random variable $0<Z<1$ controls which channel is activated; i.e. if $Z$ is lower than the threshold $\beta$, then the AWGN channel is activated, otherwise the BSC channel will be activated. 

The AWGN channel is modeled as $Y=X+N$ where $X \sim \mathcal{N} (0, \sigma_X^2)$ and $N \sim \mathcal{N} (0, \sigma_N^2)$. BSC channel is modeled as $Y = X \oplus E$, where  $X$ and $E$ are two binary random variables $X \sim \text{Bern}(p)$ and  $E \sim \text{Bern}(Z)$  denoting the input and the error respectively. This means that the probability of error in the BSC channel is controlled by the variable $Z$ at each time-point. Equivalently, if we suppose that the sample $z_i$ is observed from $Z$ at the moment $i$, the output of the BSC at time $i$ is characterized by:
\begin{equation} y_i = \left\{ \begin{array}{ll} x_i & \text{with probability } 1-z_i \\ \neg x_i & \text{with probability } z_i \end{array}\right. \end{equation}
Let's assume $Z = \min \left( \alpha, \tilde{Z} \right)$ where $\tilde{Z} \sim U(0,1)$ is an auxiliary uniform random variable, similar to the previous experiments. The theory value for $I(X:Y|Z)$ is obtained as follows:

\begin{eqnarray}
I(X;Y|Z) & = & \int_{z=0}^1 I(X;Y|Z=z) f_Z(z) dz \\
& = & \int_{z=0}^{\beta} I_{\text{AWGN}}(X;Y|Z=z) f_Z(z) dz + \int_{z=\beta}^{1} I_{\text{BSC}}(X;Y|Z=z) f_Z(z) dz  \hspace{.5in} \\
& = & \beta I_{\text{AWGN}}(X;Y) + \int_{z=\beta}^{\alpha} I_{\text{BSC}}(X;Y|Z=z) dz \nonumber \\
& & + \left(1-\alpha \right) I_{\text{BSC}}(X;Y|Z=\alpha) \\
& = & \frac{\beta}{2} \log \left( 1+ \frac{\sigma_X^2}{\sigma_N^2} \right) + \int_{z=\beta}^{\alpha} \Big( h(z \bar{p}+ p \bar{z}) - h(z) \Big) dz \nonumber \\
& & + \left(1-\alpha \right) \Big(  h(\alpha \bar{p} + p \bar{\alpha}) - h(\alpha) \Big)
\end{eqnarray}

in which $h(x)$ is the binary entropy function defined as $h(x) = - x \log x - (1-x) \log(1-x)$ for $x \in {[0,1]}$. 

In our experiment discussed in Section $\ref{sec:experiments}$ we set $p = 0.5$, $\alpha=0.3$, $\beta=0.2$, $\sigma_X=1$ and $\sigma_N=0.1$. The theory value for the conditional mutual information can be readily calculated as $I(X;Y|Z)=0.53241$. We simulated the system for various number of samples, and obtained the estimated CMI values $\hat{I}_N(X;Y|Z)$ for different methods. The results are shown in Figure \ref{fig:exp2_cmi_vs_n}. Furthermore, we calculated the estimates of $I(X;Y|Z,Z^2,Z^3)$. Its theory value is obviously equal to $I(X;Y|Z)$, yet it's conditioned over a low-dimensional manifold in a high-dimensional space, and we are interested in examining the effect of it over the estimators' accuracy. The results are shown in Figure \ref{fig:appendix_exp2_cmi_vs_n_ld}.

\subsection{Experiment 3: Total Correlation for independent mixtures}

In the third set of experiments, we estimated the total correlation of three independent random variables $X$, $Y$ and $Z$ each of which is created independently from a mixture distribution as follows: First we generate an auxiliary random variable $\tilde{X} \sim \text{Bern}(0.5)$ then the random variable $X$ is generated as follows:
\begin{equation} X = \left\{ \begin{array}{ll} \alpha_X & \text{  if  }\tilde{X}=0 \\ \sim \mathcal{U}(0,1) & \text{  if  }\tilde{X}=1  \end{array} \right. \end{equation}
which means we toss a fair coin, if heads appears we will fix $X$ at $\alpha_X$, otherwise we will draw $X$ from a uniform distribution between $0$ and $1$. samples from $Y$ and $Z$ are also generated independently in the same fashion. For this setup, the theory value of the total correlation of the three variables $X$, $Y$ and $Z$ is obviously equal to $0$.

In our experiment we set $\alpha_X=1$, $\alpha_Y=1/2$ and $\alpha_Z=1/4$, and generated various datasets with different number of samples. Then we estimated the total correlation via different approaches. The results are shown in the Figure \ref{fig:exp3_tc_vs_n}. 

\subsection{Experiment 4: Total Correlation for independent uniforms with correlated zero-inflation}

In this experiment we examine a system of random variables, which can be \textit{clustered} into independent subsets, while inside each of the subsets, the variables are dependent. As a simple case we consider $4$ random variables $X_1$, $X_2$, $X_3$ and $X_4$ which can be clustered as $\{ X_1, X_2 \}$ and $\{ X_3, X_4 \}$. Suppose $\tilde{X}_1$, $\tilde{X}_2$, $\tilde{X}_3$ and $\tilde{X}_4$ are four independent auxiliary random variables identically distributed as $\mathcal{U}(0.5,1.5)$. Then we erase samples $(\tilde{X}_1$, $\tilde{X}_2)$ and $(\tilde{X}_3$, $\tilde{X}_4)$ independently, to generate the random variables $X_1$, $X_2$, $X_3$ and $X_4$; i.e. $X_1=\alpha_1 \tilde{X}_1$, $X_2=\alpha_1 \tilde{X}_2$ and $X_3=\alpha_2 \tilde{X}_3$, $X_4=\alpha_2 \tilde{X}_4$ in which $\alpha_1 \sim \text{Bern}(p_1)$ and $\alpha_2 \sim \text{Bern}(p_2)$. As we see, after this erasure process resulting in zero-inflation, $X_1$ and $X_2$ become correlated, so do $X_3$ and $X_4$ while still $(X_1,X_2) \indep (X_3,X_4)$. Thus the total correlation can be written as:
\begin{eqnarray}
TC \left( X_1, X_2, X_3,X_4 \right) & = & h(X_1) + h(X_2) + h(X_3) + h(X_4) - h \left( X_1, X_2, X_3, X_4 \right) \\
& = & h(X_1) + h(X_2) + h(X_3) + h(X_4) - h \left( X_1, X_2 \right) - h \left( X_3, X_4 \right) \hspace{.3in} \\
& = & I \left( X_1; X_2 \right) + I \left( X_3; X_4 \right) 
\end{eqnarray}

To calculate $I \left( X_1; X_2 \right)$, after applying the chain rule to $I \left( X_1; X_2,\alpha_1 \right)$ twice, we have:
\begin{eqnarray}
I \left( X_1; X_2 \right) & = & I \left( X_1; \alpha_1 \right) + I \left( X_1; X_2 | \alpha_1 \right) - I \left( X_1;  \alpha_1 | X_2 \right) \\
& = & h(\alpha_1) -  \underbrace{ h\left( \alpha_1 \middle| X_1 \right)}_{0} \\
& & + p(\alpha_1=0) \underbrace{I \left( X_1; X_2 | \alpha_1=0 \right)}_{0} + p(\alpha_1=1) \underbrace{I \left( X_1; X_2 | \alpha_1=1 \right)}_{0} \nonumber \\
& & -h \left( X_1 \middle| X_2 \right) + \underbrace{h \left( X_1 \middle| X_2, \alpha_1 \right)}_{=h \left( X_1 \middle| X_2 \right)} \nonumber \\
& = & h(\alpha_1) 
\end{eqnarray}

Similarly, $I \left( X_3; X_4 \right)= h(\alpha_2)$. Thus:
\begin{equation} TC \left( X_1, X_2, X_3,X_4 \right) = h(\alpha_1) + h(\alpha_2)\end{equation}
 In the experiment, we set $p_1=p_2=0.6$. The theory value for the total correlation is equal to $1.34602$. The results of running different algorithms over the data can be seen in Figure \ref{fig:exp4_tc_vs_n}.

\subsection{Experiment 5: Gene Regulatory Networks}

In this experiment we use different estimators to do Gene Regulatory Network inference based on the Restricted Directed Information (RDI) measure. 

In a simplified model of a dynamical system $X=(X_1, \ldots, X_d)$, each $X_l$ is a time-series of length $T$ written in the form $\{ X_l(t) \}_{t=0}^T$, and the system's evolution through time is characterized as $X_l(t) = g_l\big( X(t-1) \big) + N_l(t)$ in which $g_l(.): \mathcal{X} \rightarrow \mathcal{X}_l$ is a deterministic function and $N_l(t)$ is an independent random noise. In the causal inference, the goal is to infer the set of  $\gls{pa_xl_var}$ for each $X_l$. To reach this end, many researchers have studied the information theoretic measures. One of such measures for time series is the directed information, and a variation of it is the restricted directed information defined as $RDI ( X_i \rightarrow X_j) := I \big( X_i(t-1), X_j(t) | X_j(t-1) \big)$ and the conditional version of it (cRDI) is defined as $RDI ( X_i \rightarrow X_j | Z ) := I \big( X_i(t-1), X_j(t) | X_j(t-1),Z(t-1) \big)$. It's shown that for the first-order markov systems and under some mild conditions,  $RDI ( X_i \rightarrow X_j | \{X_i,X_j\}^c ) \neq 0$ if and only if $X_i \in pa(X_j)$ \cite{rahimzamani2016network}. Since RDI (cRDI) is in fact a conditional mutual information, its performance is directly related to that of the CMI estimator, which is used to obtain cRDI from the samples. Hence we are interested in evaluating the performance of various estimators in causal inference.

We do our test on the simulated neuron cells' development process, based on a model from \cite{qiu2012understanding}. In this model, the time series vector $X$ consists of $13$ random variables each of which corresponding to a single gene in the development process. We simulated the development process for various values of $T$ in which the additive noise is independent and identically distributed as $\mathcal{N}(0,.03)$ for all the genes, and every single sample is then subject to erasure (i.e. be replaced by 0s) with a probability of $0.5$. Then we applied the cRDI method to the data to discover the pairwise causal relationships. In our method, we first applied the RDI method (with no conditioning) and obtained the pairwise values. Then we repeated the process to obtain cRDI values, and for every pair $X_i$ and $X_j$, we conditioned the RDI over the $X_{k^*}$ in which $k^* = \text{argmax}_{k\neq i} RDI ( X_k \rightarrow X_j)$. 

To calculate RDI and cRDI, various CMI estimators are utilized and then the Area-Under-ROC curve (AUROC) is calculated for each estimator. The more detailed results are shown in Figure \ref{fig:appendix_grn_auc_vs_n}. We notice that the zero-inflation is highly detrimental to the causal signals, so we won't expect high performance of the causal inference over the data. As we see, RDI/cRDI methods implemented with the GDM estimator outperform the other estimators by at least $\%10$ in terms of AUROC. We did not include the $\Sigma H$ estimator results due to its very low performance.


\begin{figure}
\centering
	\begin{subfigure}[t]{.45\textwidth}
	\centering
	\includegraphics[width=\textwidth,trim={0.3cm 0 1cm 0},clip]{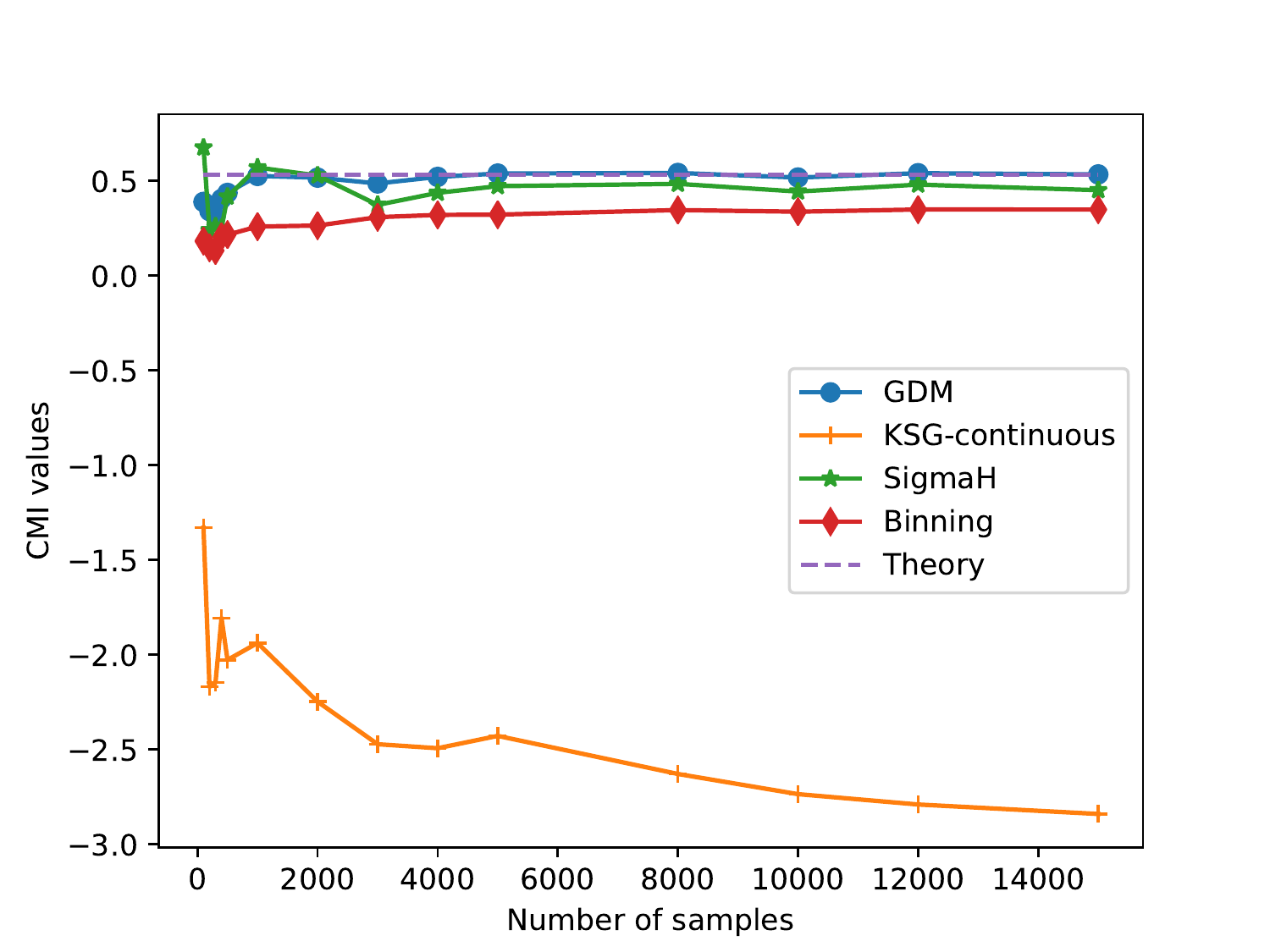}
	\caption{}
	\label{fig:exp2_cmi_vs_n}
	\end{subfigure}
	\begin{subfigure}[t]{.45\textwidth}
	\centering
	\includegraphics[width=\textwidth,trim={0.3cm 0 1cm 0},clip]{AWGN_BSC_low_dimension.pdf}
	\caption{}
	\label{fig:appendix_exp2_cmi_vs_n_ld}
	\end{subfigure}

	\begin{subfigure}[t]{.45\textwidth}
	\centering
	\includegraphics[width=\textwidth,trim={0.3cm 0 1cm 0},clip]{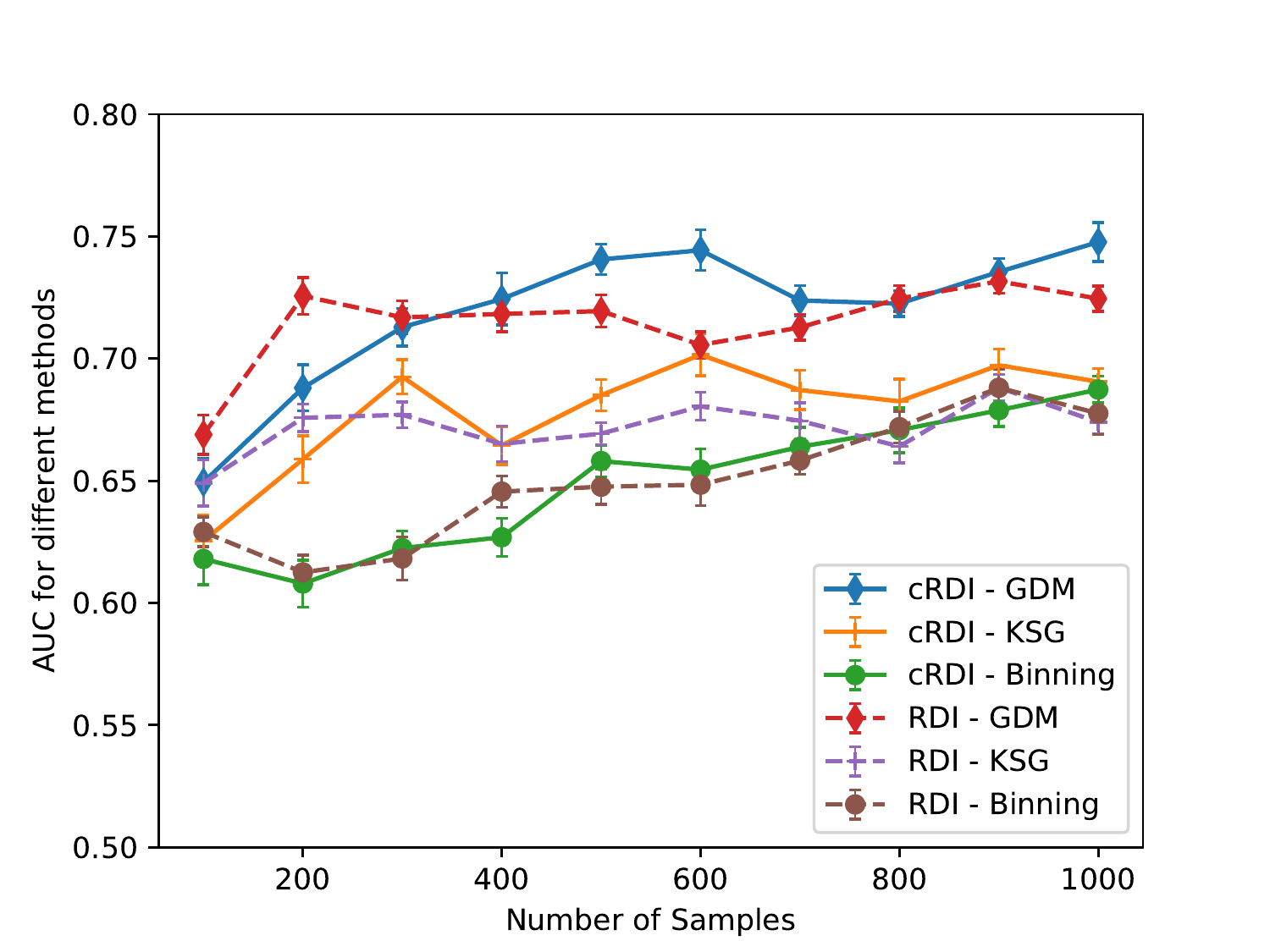}
	\caption{}
	\label{fig:appendix_grn_auc_vs_n}
	\end{subfigure}
	\begin{subfigure}[t]{.45\textwidth}
	\centering
	\includegraphics[width=\textwidth,trim={0.3cm 0 1cm 0},clip]{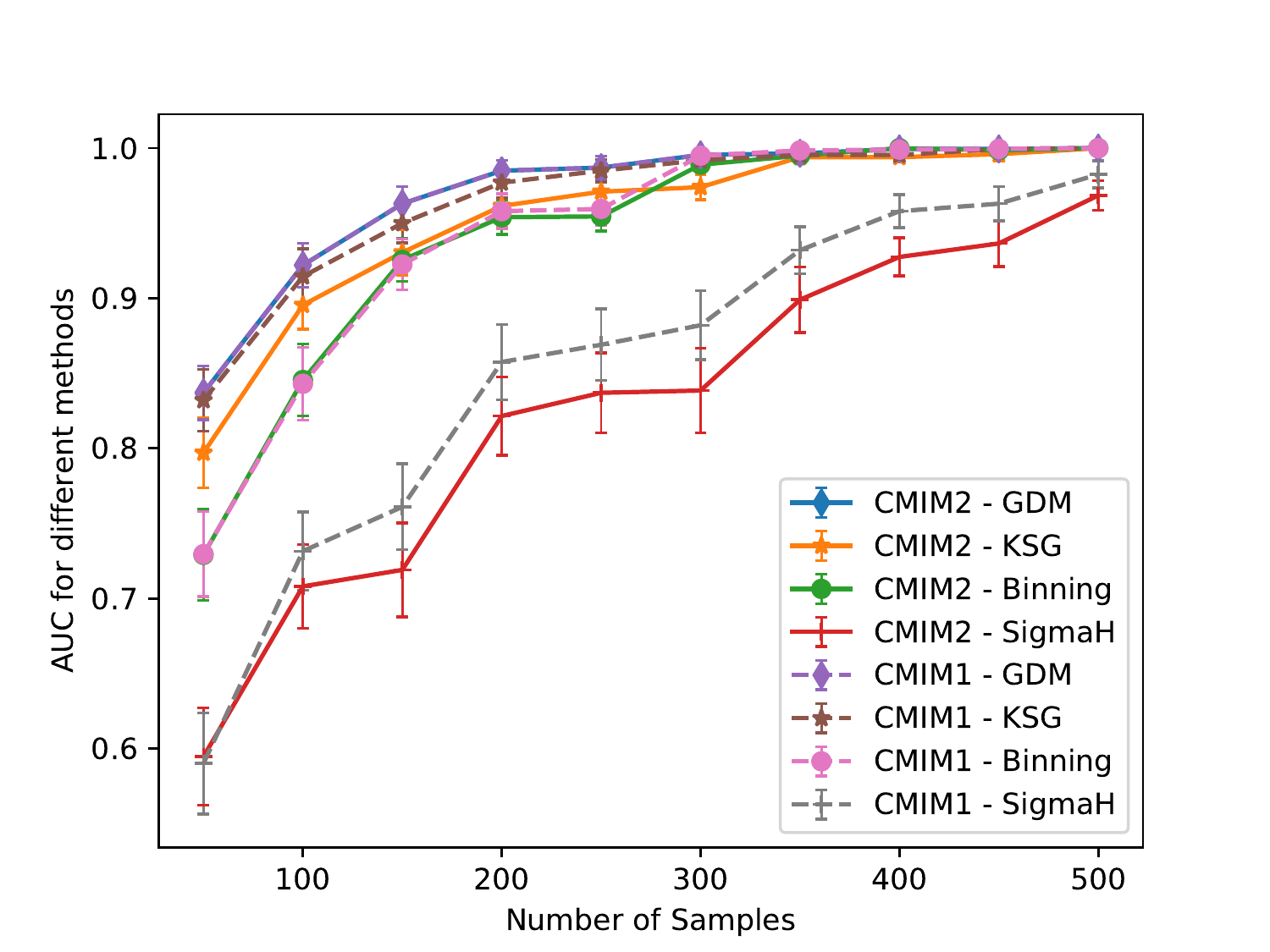}
	\caption{}
	\label{fig:appendix_featureselection_auc_vs_n}
	\end{subfigure}

\caption{The results for the experiments:  \ref{fig:exp2_cmi_vs_n}: The estimated CMI for the AWGN+BSC channels. \ref{fig:appendix_exp2_cmi_vs_n_ld}: CMI for the AWGN+BSC channels with low-dimensional $Z$ manifold. \ref{fig:appendix_grn_auc_vs_n}: The AUROC values vs the number of samples for gene regulatory network inference with different estimators. The error bars show the standard deviations scaled down by $0.2$. \ref{fig:appendix_featureselection_auc_vs_n}: The AUROC values vs the number of samples for feature selection accuracy with different estimators. The error bars show the standard deviations scaled down by $0.2$.}
\end{figure}

\subsection{Experiment 6: Feature Selection by Conditional Mutual Information Maximization}

Selecting the best features for a learning task using information theoretic measures is well studied in the literature \cite{vergara2014review}. Among the well-known methods is the conditional mutual information maximization (CMIM) first introduced by Flueret \cite{fleuret2004fast}, a variation of which was later introduced called CMIM-2 \cite{vergara2010cmim2}. Suppose we have observed samples from the data $(X,Y)$, and would like to select the set $S \subset X$ which describes $Y$ the best. Suppose we start at $S=\emptyset$ and we add features to $S$ in a recursive greedy fashion as:
\begin{eqnarray}
\text{If } S=\emptyset &:& S \leftarrow S \cup \{ \text{argmax}_{X_i} I(X_i;Y) \} \label{alg:feature_selection} \\
\text{Otherwise } &:& S \leftarrow S \cup \{ \text{argmax}_{X_i} f(X_i,S)  \} \nonumber
\end{eqnarray}
The algorithm \ref{alg:feature_selection} is equivalent to CMIM when $f(X_i,S) := \min_{X_j \in S} I(X_i,Y|X_j) $, and if we let $f(X_i,S) := \sum_{X_j} I(X_i,Y|X_j)$ we will get CMIM-2 instead.

In our experiment, we generated a vector $X=(X_1, \ldots, X_{15})$ of 15 random variables in which all the random variables are taken from $\mathcal{N}(0,1)$ and then each random variable $X_i$ is clipped from above at $\alpha_i$ which is initially taken randomly from $\mathcal{U}(0.25,0.3)$ and then kept constant during the sample generation. Then $Y$ is generated as $Y=\cos \big( \sum_{i=1}^5 X_i \big)$. So only the first 5 features are relevant and the other features are independent of $Y$. Thus an ideal feature selection method should be able to recover the first 5 features first. We implemented and ran both CMIM and CMIM-2 algorithms with various CMI estimators to evaluate the performance of the estimators and see how well they can extract the true features $X_1, \ldots, X_5$. A more detailed version of AUROC plot including both CMIM and CMIM-2 is shown in Figure \ref{fig:appendix_featureselection_auc_vs_n} . We can see that the feature selection methods implemented with the GDM estimator outperform the other estimators. We notice that the performances of CMIM and CMIM-2 with the GDM estimator are almost identical.

\end{document}